\documentclass[journal,12pt,draftclsnofoot,onecolumn]{IEEEtran}
\usepackage{ifplatform}
\usepackage{graphicx}
\usepackage{amsmath}
\usepackage{amsfonts}
\usepackage{amssymb}
\usepackage{amsthm}
\usepackage{cite}
\usepackage{pdfpages}
\usepackage{tikz,pgfplots}
\pgfplotsset{compat=newest}
\usetikzlibrary{plotmarks}
\usepackage[sans]{dsfont}
\usepackage{float}
\usepackage{stfloats}
\usepackage{topcapt}
\usepackage{xcolor}
\usepackage[justification=centering]{caption}

\newcommand{\av}{{\bf a}}
\newcommand{\bv}{{\bf b}}

\newcommand{\dv}{{\bf d}}
\newcommand{\ev}{{\bf e}}

\newcommand{\hv}{{\bf h}}

\newcommand{\sv}{{\bf s}}

\newcommand{\uv}{{\bf u}}
\newcommand{\wv}{{\bf w}}
\newcommand{\vv}{{\bf v}}
\newcommand{\xv}{{\bf x}}
\newcommand{\yv}{{\bf y}}

\newcommand{\zerov}{{\bf 0}}

\newcommand{\Am}{{\bf A}}
\newcommand{\Bm}{{\bf B}}
\newcommand{\Cm}{{\bf C}}
\newcommand{\Dm}{{\bf D}}

\newcommand{\Hm}{{\bf H}}
\newcommand{\Id}{{\bf I}}

\newcommand{\Km}{{\bf K}}
\newcommand{\Lm}{{\bf L}}
\newcommand{\Mm}{{\bf M}}

\newcommand{\Pm}{{\bf P}}
\newcommand{\Qm}{{\bf Q}}
\newcommand{\Rm}{{\bf R}}

\newcommand{\Um}{{\bf U}}
\newcommand{\Wm}{{\bf W}}
\newcommand{\Vm}{{\bf V}}
\newcommand{\Xm}{{\bf X}}
\newcommand{\Ym}{{\bf Y}}
\newcommand{\Zm}{{\bf Z}}

\newcommand{\Ac}{{\cal A}}

\newcommand{\Cc}{{\cal C}}

\newcommand{\Ic}{{\cal I}}

\newcommand{\Mc}{{\cal M}}
\newcommand{\Nc}{{\cal N}}
\newcommand{\Oc}{{\cal O}}

\newcommand{\Sc}{{\cal S}}

\newcommand{\Hcb}{\pmb{\cal H}}

\newtheoremstyle{mystyle}
{}
{}
{\itshape}
{}
{\bf}
{}
{.5em}
{\thmname{#1}\thmnumber{ #2}. \textnormal{\thmnote{(#3)}}}

\theoremstyle{mystyle}
\newtheorem{mylem}{Lemma}
\newtheorem{myexm}{Example}
\newtheorem{mycor}{Corollary}
\newtheorem{mythm}{Theorem}

\allowdisplaybreaks
\begin{document}
\title{On Low-Complexity Full-diversity Detection In Multi-User MIMO Multiple-Access Channels}
\author{\IEEEauthorblockN{Amr~Ismail,~\IEEEmembership{Member, IEEE} and~Mohamed-Slim~Alouini,~\IEEEmembership{Fellow, IEEE}}}

\maketitle
\begin{abstract}
Multiple-input multiple-output (MIMO) techniques are becoming commonplace in recent wireless communication standards. This added dimension (i.e., space) can be efficiently used to mitigate the interference in the multi-user MIMO context. In this paper, we focus on the uplink of a MIMO multiple access channel (MAC) where perfect channel state information (CSI) is only available at the destination. We provide a new set of sufficient conditions for a wide range of space-time block codes (STBC)s to achieve full-diversity under \emph{partial interference cancellation group decoding} (PICGD) with or without successive interference cancellation (SIC) for completely blind users. Explicit interference cancellation (IC) schemes for two and three users are then provided and shown to satisfy the derived full-diversity criteria. Besides the complexity reduction due to the fact that the proposed IC schemes enable separate decoding of distinct users without sacrificing the diversity gain, further reduction of the decoding complexity may be obtained. In fact, thanks to the structure of the proposed schemes, the real and imaginary parts of each user's symbols may be decoupled without any loss of performance. Finally, our theoretical claims are corroborated by simulation results and the new IC scheme for two-user MIMO MAC is shown to outperform the recently proposed two-user IC scheme especially for high spectral efficiency while requiring significantly less decoding complexity.   
\end{abstract}

\let\thefootnote\relax\footnotetext{The authors are with the Computer, Electrical, and Mathematical Sciences and Engineering (CEMSE) Division, King Abdullah University of Science and Technology (KAUST), Thuwal, Makkah Province, Kingdom of Saudi Arabia (e-mail:$\lbrace$amrismail.tammam, slim.alouini$\rbrace$@kaust.edu.sa)}

\section{Introduction}
Interference mitigation is a major issue in the design of wireless communication systems. Classical techniques to cancel the interference rely on sharing the available resources among the different users. In time division multiple access (TDMA) systems, the time slots are divided between different users such that in each time slot one user is transmitting solely. Similarly, in frequency division multiple access (FDMA) systems, the frequency band is divided between the different users such that the users are transmitting their data through disjoint sub-carriers. A higher number of users may be accommodated if a hybrid scheme of TDMA/FDMA is adopted which is the case for the global system for mobile communications (GSM) where eight users are time division multiplexed over each sub-carrier. In a system that employs direct sequence code division multiple access (DS-CDMA), different users share the same frequency band but are assigned orthogonal spread sequences, such that the receiver is capable of extracting a specific user's data through correlating the received signal with the corresponding spreading sequence. 

The incorporation of multiple-input multiple-output (MIMO) schemes in recent standards such as Wireless Personal Area Networks \cite{IEEE_802.15c}, Wireless Local Area Networks \cite{IEEE_802.11n}, mobile WIMAX \cite{IEEE_802.16e}, 3GPP LTE (release 8 and 9), and LTE-advanced (release 10), paved the way to exploit the added dimension (i.e., space) to efficiently cancel the interference without requiring additional resources. In fact, \emph {space-based} interference cancellation (IC) techniques can be employed along with conventional techniques (e.g., TDMA, FDMA, CDMA, etc.) to increase the number of accommodated users.  

When designing the so-called space-based interference cancellation techniques, several objectives have to be taken into consideration, namely, providing high-rate communication, achieving the full-diversity, the simplicity of the decoding algorithm measured in terms of average and worst-case decoding complexity \cite{FAST} and finally, the ability to accommodate a high number of users. In this context, full-diversity denotes the maximal diversity gain offered by the network configuration as if each user was transmitting solely, in other words, the \emph{interference-free} maximal diversity gain. Towards this end, several interference cancellation schemes have been proposed in the literature. In what follows, the existing interference cancellation schemes will be outlined.


\subsection*{Prior Work}
In \cite{2-USER_MIMO_MAC}, the authors proposed a scheme to suppress the interference for a two-user MIMO multiple access channel (MAC) where each user is equipped with two transmit antennas and the common receiver has $r+1$ antennas. The  proposed scheme assumed blind transmitters and relied on the orthogonality of the Alamouti codewords \cite{ALAMOUTI}. This scheme enables a rate-1 symbol per channel use (spcu) communication for each user with  fixed detection complexity of $\Oc(1)$ irrespective of the underlying quadrature amplitude modulation (QAM) constellation. The main drawback of this scheme is the substantial loss of diversity, in fact the achievable diversity gain with this scheme is $2r$ instead of $2\left(r+1\right)$. The problem of diversity loss was partially resolved in \cite{2-USER_MIMO_MAC_FD}, where the authors proposed a decoding algorithm that achieves full-diversity in the two-user MIMO MAC setting provided that the destination is equipped with more than two antennas. Latter, the original scheme in \cite{2-USER_MIMO_MAC} was extended in \cite{J-USER_MIMO_MAC_2TX} to the case of $J$ users each equipped with two transmit antennas. However the provided diversity gain assuming a receiver equipped with $J+r-1$ antennas is equal to $2r$. Subsequently, this scheme was generalized in \cite{J-USER_MIMO_MAC_NTX} for the case of $J$ users with $N$ transmit antennas, where for a receiver equipped with $J+r-1$ antennas, the achieved diversity gain is equal to $Nr$.

In order to overcome the diversity gain loss inherent in the previous schemes, F. Li and H. Jafarkhani proposed a full-diversity cancellation scheme for two users each equipped with $2N$ transmit antennas and a receiver equipped with $M$ receive antennas \cite{2-USER_MIMO_MAC_FD_CSIT}. The proposed scheme relies on creating orthogonal subspaces at the receiver for the different users, thus retaining the maximum diversity gain under the assumption of global channel state information (CSI) availability at the transmitters with a worst-case decoding complexity of $\Oc\left(q^N\right)$, where $q$ denotes the size of the underlying square QAM constellation. The full knowledge of the global CSIT was then relaxed in \cite{2-USER_MIMO_MAC_FD_LFB}, where only a limited feedback is available to the transmitters. Moreover, in \cite{J-USER_MIMO_MAC_QOSTBC}, the authors extended the IC scheme in \cite{2-USER_MIMO_MAC_FD_CSIT} to accommodate more users but under the constraint of the availability of global CSIT. Recently, in \cite{2-USER_MIMO_MAC_PIC}, the authors proposed a systematic full-diversity IC scheme with \emph{partial interference cancellation group decoding} (PICGD) for two blind users with asymptotic individual rate of $1$ spcu when the signalling period gets too large w.r.t the number of transmit antennas.  
 

\subsection*{Our Contribution}
In this paper, we address the design of low-complexity, full-diversity multi-user space-time coding schemes for the MIMO MAC uplink. The outcome of the present paper may be summarized as follows.
\begin{itemize}
\item A new set of sufficient conditions for a wide range of space-time block codes (STBC)s to achieve full-diversity under PICGD \cite{PICGD,PICGD_errata} with or without successive interference cancellation (SIC) for completely blind users is derived. Compared to the derived sufficient conditions in \cite{2-USER_MIMO_MAC_PIC}, our design criteria is less restrictive, thus enabling higher rate IC schemes. In fact, as will be shown shortly, the design criteria in \cite{2-USER_MIMO_MAC_PIC} need to be satisfied only {\bf \emph{almost surely}} rather than with strict equality.
\item Explicit IC schemes for two and three users satisfying the derived design criteria under PICGD and PICGD-SIC, respectively are provided.
\item Besides the ability of the proposed IC schemes to decode distinct users disjointly without sacrificing the full-diversity, it is proven (see Appendices B, and C) that further reduction of the decoding complexity may be obtained through separate decoding of the real and imaginary parts of each user's symbols without incurring any loss of performance.
\end{itemize}
Our theoretical claims are first corroborated by simulation results, i.e., the diversity gain of our two and three users IC schemes are numerically verified to be identical to the full-diversity gain under the configuration of interest. The proposed two-user IC scheme is then compared to its counterpart in \cite{2-USER_MIMO_MAC_PIC} in terms of the codeword error rate (CER), the average decoding complexity, and the worst-case decoding complexity. It is found that our scheme outperforms its counterpart in \cite{2-USER_MIMO_MAC_PIC} especially for high spectral efficiencies, while requiring significantly less decoding complexity. 


\subsection*{Notations}
Throughout the paper, small letters, bold small letters, bold capital letters, and calligraphic letters will designate scalars, vectors, matrices, and sets respectively. If $\Am$ is a matrix, then $\Am^\mathsf{H}$, $\Am^\mathsf{T}$, $\Am^{\dagger}$, and $\text{r}\left(\Am\right)$ denote the Hermitian, the transpose, the pseudo-inverse, and the rank of $\Am$, respectively. We define $\Mc\left(\Am\right)$ to be the vector space spanned by the columns of $\Am$ and the $\text{vec}(\Am)$ as the operator which  when applied to a $m \times n$ matrix $\Am$, transforms it  into a $mn\times 1$ vector by vertically concatenating  the columns of the corresponding matrix. The $\text{diag}\left(\av\right)$ operator returns a square matrix with the vector $\av$ on its main diagonal. For two column vectors $\av$ and $\bv$, $\left<\av,\bv\right>$ denotes their dot product $\left(\av^\mathsf{H}\bv\right)$. If a matrix $\Am\succeq 0$, then $\Am$ is positive semi-definite. The $\otimes$ operator is the Kronecker product. The notion $\Pm_{\pmb{\sigma}}$ denotes a $n\times n$ permutation matrix defined as $\begin{bmatrix}\ev_{\pmb{\sigma}\left(1\right)}&\ev_{\pmb{\sigma}\left(2\right)}&\ldots&\ev_{\pmb{\sigma}\left(n\right)}\end{bmatrix}$,
where $\ev_i$ denotes the $i$-th column of the $n\times n$ identity matrix, and $\pmb{\sigma}\in\Sc_n$; the set of all permutations over $\left\lbrace1,\ldots,n\right\rbrace$. The $\Re\left(.\right)$ and $\Im\left(.\right)$ operators denote the real and imaginary parts of their arguments, respectively. $\Id_n$ denotes the $n\times n$ identity matrix, while $\zerov$ denote the null matrix of appropriate size. If a random vector $\xv\sim\Cc\Nc\left(\pmb{\mu},\Km\right)$, then $\xv$ is drawn from a circularly symmetric complex Gaussian distribution with mean $\pmb{\mu}$ and covariance matrix $\Km$. Finally, $\mathbb{P}\left[A\right]$ denotes the probability of the event $A$ to occur, and $\mathbb{E}_x\left[f\left(x\right)\right]$ denotes the statistical average of an arbitrary function $f\left(x\right)$ w.r.t the random variable $x$.


\section{System model}

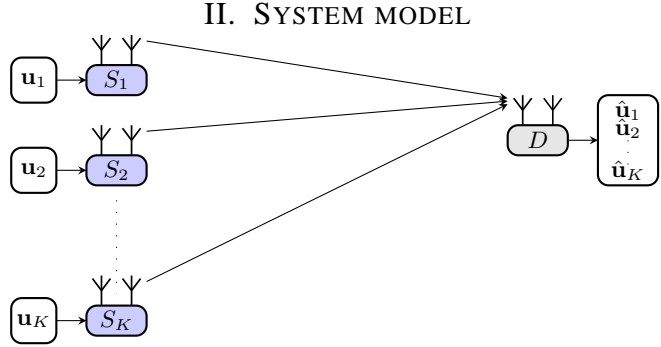
\begin{figure}[h!]
\centering 
\begin{tikzpicture}[x=0.4cm,y=0.4cm]
\filldraw [fill=blue!20,rounded corners,thick] (0,4) rectangle (2,5);
\draw[-,line width=0.65pt] (0.5,5) -- (0.5,6);
\draw[-,line width=0.65pt] (0.5,5.5) -- (0.8,5.9);
\draw[-,line width=0.65pt] (0.5,5.5) -- (0.2,5.9);
\draw[-,line width=0.65pt] (1.5,5) -- (1.5,6);
\draw[-,line width=0.65pt] (1.5,5.5) -- (1.8,5.9);
\draw[-,line width=0.65pt] (1.5,5.5) -- (1.2,5.9);
\draw (1,4.5) node {\footnotesize $S_1$};
\filldraw [fill=blue!20,rounded corners,thick] (0,1) rectangle (2,2);
\draw[-,line width=0.65pt] (0.5,2) -- (0.5,3);
\draw[-,line width=0.65pt] (0.5,2.5) -- (0.8,2.9);
\draw[-,line width=0.65pt] (0.5,2.5) -- (0.2,2.9);
\draw[-,line width=0.65pt] (1.5,2) -- (1.5,3);
\draw[-,line width=0.65pt] (1.5,2.5) -- (1.8,2.9);
\draw[-,line width=0.65pt] (1.5,2.5) -- (1.2,2.9);
\draw (1,1.5) node {\footnotesize $S_2$};
\draw[-,loosely dotted,thin] (1,0.5) -- (1,-2.6);
\filldraw [fill=blue!20,rounded corners,thick] (0,-4) rectangle (2,-3);
\draw[-,line width=0.65pt] (0.5,-3) -- (0.5,-2);
\draw[-,line width=0.65pt] (0.5,-2.5) -- (0.8,-2.1);
\draw[-,line width=0.65pt] (0.5,-2.5) -- (0.2,-2.1);
\draw[-,line width=0.65pt] (1.5,-3) -- (1.5,-2);
\draw[-,line width=0.65pt] (1.5,-2.5) -- (1.8,-2.1);
\draw[-,line width=0.65pt] (1.5,-2.5) -- (1.2,-2.1);
\draw (1,-3.5) node {\footnotesize $S_K$};
\draw[-stealth,solid,thin] (2,5.8) -- (14,3.9);
\draw[-stealth,solid,thin] (2,2.8) -- (14,3.8);
\draw[-stealth,solid,thin] (2,-2.2) -- (14,3.7);

\filldraw [fill=white,rounded corners,thick] (-2.5,3.75) rectangle (-1,5.25);
\draw[-stealth,solid,thin] (-1,4.5) -- (0,4.5);
\draw (-1.75,4.5) node {\footnotesize $\uv_1$};
\filldraw [fill=white,rounded corners,thick] (-2.5,0.75) rectangle (-1,2.25);
\draw[-stealth,solid,thin] (-1,1.5) -- (0,1.5);
\draw (-1.75,1.5) node {\footnotesize $\uv_2$};
\filldraw [fill=white,rounded corners,thick] (-2.5,-4.25) rectangle (-1,-2.75);
\draw[-stealth,solid,thin] (-1,-3.5) -- (0,-3.5);
\draw (-1.75,-3.5) node {\footnotesize $\uv_K$};

\filldraw [fill=gray!20,rounded corners,thick] (14,2) rectangle (16,3);
\draw[-,line width=0.65pt] (14.5,3) -- (14.5,4);
\draw[-,line width=0.65pt] (14.5,3.5) -- (14.8,3.9);
\draw[-,line width=0.65pt] (14.5,3.5) -- (14.2,3.9);
\draw[-,line width=0.6pt] (15.5,3) -- (15.5,4);
\draw[-,line width=0.6pt] (15.5,3.5) -- (15.8,3.9);
\draw[-,line width=0.6pt] (15.5,3.5) -- (15.2,3.9);
\draw (15,2.5) node {\footnotesize $D$};

\filldraw [fill=white,rounded corners,thick] (17,1) rectangle (19,4);
\draw[-stealth,solid,thin] (16,2.5) -- (17,2.5);
\draw (18,3.5) node {\footnotesize $\hat{\uv}_1$};
\draw (18,2.9) node {\footnotesize $\hat{\uv}_2$};
\draw[-,loosely dotted,thin] (18,2.5) -- (18,1.6);
\draw (18,1.5) node {\footnotesize $\hat{\uv}_K$};

\end{tikzpicture}
\label{system_model}
\caption{$K$-User MIMO MAC channel uplink}
\end{figure}

Suppose that we have $K$ sources equipped each with $N_t$ transmit antennas and a destination equipped with $N_r$ receive antennas, we will describe this configuration afterwards by $\left(N_t^K,N_r\right)$. The baseband MIMO MAC uplink channel may be described by
\begin{equation}
\underset{T\times N_r}{\Ym} =\sum^{K}_{k=1}\underset{T\times N_t}{\Xm_k} \underset{N_t\times N_r}{\Hm_k} +\underset{T\times N_r}{\Wm},
\label{MIMO-MAC}
\end{equation} 
where $T$ is the codeword signalling period, $N_r$ is the number of receive antennas, $N_t$ is the number of transmit antennas, $\Ym$ is the received signal matrix, and $\Xm_{km}$ is the space-time mapping of the $k$-th source $\left(S_k\right)$ message addressed to the destination $\left(D\right)$ (i.e., $\Xm_k=f\left(\uv_k\right)$, where $f\left(.\right)$ is an injective function for the code to be uniquely decodable, which completely defines the coding scheme). The channel coefficients matrix from the $k$-th source to the common destination is denoted $\Hm_k$ whose entries $h_{ij} \sim \Cc \Nc(0,1)$, and $\Wm$ is the noise matrix at the destination with entries $w_{ij} \sim \Cc \Nc(0,N_{0})$. According to the above model, the $t$-th row of $\Xm_k$ denotes the symbols transmitted through the $N_t$ transmit antennas during the $t$-th channel use while the $n$-th column denotes the symbols transmitted through the $n$-th transmit antenna during the codeword signalling period $T$. We assume perfect synchronization among distinct users in the network and perfect CSI is available only at the destination. 

Applying the $\text{vec}(.)$ operator to both sides of \eqref{MIMO-MAC}, we obtain
\begin{equation}
\underbrace{\text{vec}\left(\Ym\right)}_{\yv} =\sum^{K}_{k=1}\underbrace{\Id_{Nr}\otimes\Xm_k}_{\widetilde{\Xm}_k}
\underbrace{\text{vec}\left(\Hm_k\right)}_{\hv_k}+
\underbrace{\text{vec}\left(\Wm\right)}_{\wv}.
\label{MAC}
\end{equation} 
For a large class of STBCs (\cite{DAST, TAST, PERFECT_STBC}, among others), the code matrices take the form
\begin{equation}
\Xm_k=\sum^{n_k}_{i=1} \Am_{k,i}s_{k,i},\ \forall\ k=1,\dots,K,
\label{STBC}
\end{equation}
with $s_{k,i}\in\mathbb{C}$ and the $\Am_{k,i}\in\mathbb{C}^{T \times N_t}$. Replacing $\Xm_k$ by its expression in Eq.~\eqref{STBC}, the system model in Eq.~\eqref{MAC} becomes
\begin{equation*}
\yv=\sum^{K}_{k=1}\sum^{n_k}_{i=1}\left(\Id_{N_r}\otimes\Am_{k,i}\right)
\hv_ks_{k,i}+\wv,
\label{vec}
\end{equation*}
or equivalently
\begin{equation}
\yv=\sum^{K}_{k=1}\widetilde{\Hcb}_k\sv_k+\wv,
\label{eq_model}
\end{equation}
with $\widetilde{\Hcb}_k=\begin{bmatrix}\Hcb^\mathsf{T}_{k,1}&\ldots& \Hcb^\mathsf{T}_{k,N_r}\end{bmatrix}^\mathsf{T}$, $\Hcb_{k,m}=\begin{bmatrix}\Am_{k,1}\hv_{k,m}&\ldots&\Am_{k,n_k}\hv_{k,m}\end{bmatrix}$, $\hv_{k,m}$ denotes the channel coefficients vector from the $k$-th source to the $m$-th receive antenna at the destination, and $\sv_k=\begin{bmatrix}s_{k,1}&\ldots&s_{k,n_k}\end{bmatrix}^\mathsf{T}$.


\subsection*{Interference Cancellation}
The objective herein is to decode distinct users data at a low-complexity cost while achieving the full-diversity gain offered by the system configuration (i.e., $N_tN_r$). For this purpose, we adopt a multi-user variant of the full-diversity PICGD \cite{2-USER_MIMO_MAC_PIC_FD} originally proposed for the case of point-to-point communications \cite{PICGD, PICGD_errata}. The PICGD is a decoding algorithm that generalizes the zero-forcing receiver, namely, it separates the transmitted symbols into disjoint sets and decodes these sets independently. For instance, to decode the $l$-th set of transmitted symbols, the receiver projects the received signal into the subspace orthogonal to the one spanned by the rest of the symbols. In \cite{PICGD,PICGD_errata}, the authors derived sufficient conditions for a STBC with a given grouping scheme to achieve full-diversity under PICGD in a point-to-point scenario. 

In the multi-user MIMO MAC setting, each symbols set will correspond to a given user's data. Suppose that we want to decode the $l$-th user's transmitted symbols, towards this end let the system model in Eq.~\eqref{eq_model} be written in the following form
\begin{equation}
\yv=\widetilde{\Hcb}_l\sv_l+\sum^{K}_{k=1,k\neq l}\widetilde{\Hcb}_k\sv_k+\wv.
\label{PIC}
\end{equation}
Therefore, the destination projects the received signal into the subspace orthogonal to the one spanned by interfering users. Let $\widetilde{\Hcb}_{\overline{l}}$ denote a basis of $\Mc\left(\begin{bmatrix}\widetilde{\Hcb}_1&\ldots&\widetilde{\Hcb}_{l-1} &\widetilde{\Hcb}_{l+1}&\ldots&\widetilde{\Hcb}_K\end{bmatrix}\right)$. Therefore, the required projection matrix $\Pm_l$ needs to satisfy $\Pm_l\widetilde{\Hcb}_{\overline{l}}=\zerov$.
This condition has a general solution described by
\begin{equation*}
\Pm_l=\Qm_l\Mm_l;\ \Mm_l=\left(\Id_{N_rT}-\widetilde{\Hcb}_{\overline{l}}\widetilde{\Hcb}_{\overline{l}}^{\dagger}\right).
\end{equation*}
It has been proved in \cite{PICGD} that taking $\Qm_l=\Id$ minimizes the maximum likelihood (ML) decoding probability of error. Hence, hereafter, we will take $\Pm_l=\left(\Id_{N_rT}-\widetilde{\Hcb}_{\overline{l}}\widetilde{\Hcb}_{\overline{l}}^{\dagger}\right)$. Left multiplying Eq.~\eqref{PIC} by $\Pm_l$ one obtains
\begin{equation*}
\Pm_l\yv=\Pm_l\widetilde{\Hcb}_l\sv_l+\Pm_l\wv.
\end{equation*}
The ML estimate of $\sv_l$ under PICGD is then given by
\begin{equation}
\sv^{\text{ML$\mid$PICGD}}_l=\text{arg}\underset{\sv_l\in\Ac_l}{\min}\Vert\Pm_l\yv-\Pm_l\widetilde{\Hcb}_l\sv_l\Vert
\label{PIC_DEC}
\end{equation} 
where $\Ac_l$ denotes the codebook spanned by $\sv_l$. For the considered class of STBCs, one has $\widetilde{\Xm}_l\hv_l=\widetilde{\Hcb}_l\sv_l$, therefore the ML estimate under PICGD may be re-written as
\begin{equation}
\widetilde{\Xm}^{\text{ML$\mid$PICGD}}_l=\text{arg}\underset{\widetilde{\Xm}_l\in\widetilde\Cc_l}{\min}\Vert\Pm_l\yv-\Pm_l\widetilde{\Xm}_l\hv_l\Vert
\label{PIC_DEC_equiv}
\end{equation} 
where $\widetilde{\Cc}_l$ denotes the codebook spanned by $\widetilde{\Xm}_l$. 

It is well known that the performance of PICGD may be significantly enhanced if combined with successive interference cancellation. At each stage, the contribution of the decoded set of symbols is subtracted from the received signal, thus reducing the dimension of the column space spanned by the interference successively. In this case, assuming that the symbols sets are enumerated with descending order of signal-to-noise ratio, the decoding process is as previously described with the only difference that at the $l$-th step, $\widetilde{\Hcb}_{\overline{l}}$ will denote a basis of $\Mc\left(\begin{bmatrix}\widetilde{\Hcb}_{l+1}&\ldots&\widetilde{\Hcb}_{K}\end{bmatrix}\right)$. 


\section{Full-diversity criteria}
In what follows, we will derive new sufficient conditions for the set of STBCs in \eqref{STBC} to achieve full-diversity under PICGD and PICGD-SIC group decoding in the multi-user MIMO MAC setting. Now, we proceed towards our main theorem.
\begin{mythm}
The STBCs $\Xm_l$ expressed as in \eqref{STBC}
\begin{equation*}
\Xm_l=\sum^{n}_{i=1}\Am_{l,i}s_{l,i},\ s_{k,i}\in\mathbb{C},\ \forall\ l=1,\ldots,K
\end{equation*}
achieve full-diversity under PICGD with grouping scheme ${\sv_1,\ldots,\sv_K}$ in the multi-user MIMO MAC setting if
\begin{itemize}
\item $\text{r}\left(\Pm_l\Delta\widetilde{\Xm}_l\right)\neq 0,\forall\ \Pm_l,\ \Delta\widetilde{\Xm}_l\in\Delta\widetilde{\Cc}_l\setminus \left\lbrace \zerov\right\rbrace,\ l=1,\ldots,K$,
\item The matrix $\Pm_l\Delta\widetilde{\Xm}_l$ is of full column rank {\bf \emph{almost surely}}, $\forall\ \Delta\widetilde{\Xm}_l\in\Delta\widetilde{\Cc}_l\setminus \left\lbrace \zerov\right\rbrace,\ l=1,\ldots,K$,
\end{itemize}
where $\Delta \widetilde{\Xm}_l$ (resp. $\Delta\widetilde{\Cc}_l$) denotes the codeword difference (resp. codeword difference codebook) of the $l$-th user.
\end{mythm}
It is worth noting that the second condition implies that $\Delta\Xm_l$ is of full column rank $ \forall\  \Delta\Xm_l\in\Delta\Cc_l\setminus\left\lbrace\zerov\right\rbrace$, in other words, achieving the full-diversity under ML for a certain STBC is a prerequisite to achieve the full-diversity under PICGD as expected. Moreover, the second condition in Theorem 1 implies that the matrix $\Pm_l\Delta\widetilde{\Xm}_l$ is of full column rank {\bf \emph{almost surely}}, which is significantly less restrictive than the provided sufficient conditions by the authors in \cite{2-USER_MIMO_MAC_PIC}, where this matrix is required to be of full column rank $\forall\ \Pm_l,\ \Delta\widetilde{\Xm}_l\in\Delta\widetilde{\Cc}_l\setminus\left\lbrace\zerov\right\rbrace$. This enables us to obtain higher rates IC schemes as will be shortly shown.
\begin{proof} 
A MIMO system is said to achieve a diversity gain $d$ if the probability of error $\text{P}_e$ can be upper-bounded in the high SNR regime as
\begin{equation}
\text{P}_e\lesssim\alpha\ \text{SNR}^{-d}
\end{equation}
where $\alpha$ is a positive constant. According to \eqref{PIC_DEC_equiv}, the conditional pairwise error probability (PEP) may be obtained as in \cite{FWC}
\begin{equation}
\mathbb{P}\left[\widetilde{\Xm}_l\rightarrow\widetilde{\Xm}_l'\mid\Pm_l,\hv_l\right]=Q\left(\sqrt{\frac{\text{SNR}\Vert\Pm_l\Delta\widetilde{\Xm}_l\hv_l\Vert^2}{2}}\right)
\end{equation}
where $\Delta\widetilde{\Xm}_l=\widetilde{\Xm}_l-\widetilde{\Xm}_l'$. Thanks to the independence between $\Pm_l$ and $\hv_l$ for $l=1,\ldots,K$, the average PEP can be evaluated in two steps, namely by averaging over the distribution of $\hv_l$ for a fixed $\Pm_l$ followed by averaging over the distribution of $\Pm_l$.
Towards this end, the conditional expectation of the PEP can be expressed as
\begin{equation}
\mathbb{P}\left[\widetilde{\Xm}_l\rightarrow\widetilde{\Xm}_l'\mid\Pm_l\right]=\mathbb{E}_{\hv_l\mid\Pm_l}\left[Q\left(\sqrt{\frac{\text{SNR}\ \widetilde{\hv}_l^\mathsf{H}\pmb{\Lambda}_l\widetilde{\hv}_l}{2}}\right)\right]
\end{equation}
where $\widetilde{\hv}_l=\Vm_l\hv_l$, $\Vm_l$ is unitary, $\pmb{\Lambda}_l=\text{diag}\left(\lambda_{l,1}^2,\ldots,\lambda_{l,N_tN_r}^2\right)$, and the $\lambda_{l,i}$'s denote the singular values of $\Pm_l\Delta\widetilde{\Xm}_l$. Let $r_l$ denote the rank of $\Pm_l\Delta\widetilde{\Xm}_l$, therefore, the above equation reduces to
\begin{equation}
\begin{split}
\mathbb{P}\left[\widetilde{\Xm}_l\rightarrow\widetilde{\Xm}_l'\mid\Pm_l\right]
&=\mathbb{E}_{\hv_l\mid\Pm_l}\left[Q\left(\sqrt{\frac{\text{SNR}\sum_{i=1}^{r_l}\lambda_{l,i}^2\vert\widetilde{\hv}_l(i)\vert^2}{2}}\right)\right]\\
&\leq\mathbb{E}_{\hv_l\mid\Pm_l}\left[\exp\left(-\frac{\text{SNR}\sum_{i=1}^{r_l}\lambda_{l,i}^2\vert\widetilde{\hv}_l(i)\vert^2}{4}\right)\right].
\end{split}
\end{equation}
Thanks to the independence between $\Pm_l$ and $\hv_l$, conditioning on $\Pm_l$ does not affect the distribution of $\hv_l$, and since $\Vm_l$ depends on $\Pm_l$, hence fixed unitary matrix, one has $\widetilde{\hv}_l\sim\Cc\Nc\left(\zerov,\Id_{N_tN_r}\right)$. Consequently, the above inequality reduces to
\begin{equation}
\mathbb{P}\left[\widetilde{\Xm}_l\rightarrow\widetilde{\Xm}_l'\mid\Pm_l\right]\leq\prod^{r_l}_{i=1}\frac{1}{1+\frac{\text{SNR}\ \lambda_{l,i}^2}{4}}
\end{equation}  
which in the high SNR regime simplifies to
\begin{equation}
\mathbb{P}\left[\widetilde{\Xm}_l\rightarrow\widetilde{\Xm}_l'\mid\Pm_l\right]\lesssim
\left(\frac{4}{\text{SNR}}\right)^{r_l}\frac{1}{\prod^{r_l}_{i=1}\lambda_{l,i}^2}.
\label{high_SNR}
\end{equation} 
Finally, the average PEP is obtained as
\begin{equation}
\mathbb{P}\left[\widetilde{\Xm}_l\rightarrow\widetilde{\Xm}_l'\right]\lesssim\mathbb{E}_{\Pm_l}\left[\left(\frac{4}{\text{SNR}}\right)^{r_l}\frac{1}{\prod^{r_l}_{i=1}\lambda_{l,i}^2}\right].
\end{equation}
On the other hand, the first condition of Theorem 1 implies that $r_l\neq 0,\ \forall\ \Pm_l,\ \Delta\widetilde{\Xm}_l\in\Delta\widetilde{\Cc}_l\setminus \left\lbrace \zerov\right\rbrace,\ l=1,\ldots,K$, hence, the second term of the R.H.S in the above inequality has a finite second moment and the Cauchy-Shwarz inequality can be applied to obtain
\begin{equation}
\mathbb{P}\left[\widetilde{\Xm}_l\rightarrow\widetilde{\Xm}_l'\right]\lesssim\sqrt{\mathbb{E}_{r_l}\left[\left(\frac{4}{\text{SNR}}\right)^{2r_l}\right]\mathbb{E}_{\pmb{\Lambda}_l}\left[\frac{1}{\prod^{r_l}_{i=1}\lambda_{l,i}^4}\right]}.
\end{equation}   
If $\Pm_l\Delta\Xm_l$ is of full column rank {\bf \emph{almost surely}}, or equivalently $r_l\overset{a.s.}{=}N_tN_r$, where $\overset{a.s.}{=}$ denotes \emph{equals almost surely}, the above inequality can be re-written as
\begin{equation}
\mathbb{P}\left[\widetilde{\Xm}_l\rightarrow\widetilde{\Xm}_l'\right]\lesssim
\alpha\left(\frac{4}{\text{SNR}}\right)^{N_tN_r}
\end{equation}
where $\alpha$ is a positive finite number. 
\end{proof}
The second condition of Theorem 1 can be easily checked thanks to the following lemma, which is a generalized form of the rank equality in \cite{MATRICES_PROD} 
\begin{mylem}
For the two matrices $\Am \in \mathbb{C}^{p\times q}$ and $\Cm \in \mathbb{C}^{p\times m}$ such that:
\begin{IEEEeqnarray}{c;c;c}
\text{r}\left(\Am\right)+\text{r}\left(\Cm\right)&=&p \label{cond1}\\
\Cm^\mathsf{H}\Am &=& \zerov \label{cond2},
\end{IEEEeqnarray}
and $\Am$ is of full column rank, we have:
\begin{equation*}
\text{r}\left(\begin{bmatrix}\Am & \Vm \end{bmatrix}\right)=\text{r}\left(\Cm^\mathsf{H}\Vm\Cm\right)+\text{r}\left(\Am\right).   
\end{equation*} 
where $\Vm \in \mathbb{C}^{p\times p}$ is an arbitrary positive semi-definite matrix. 
\end{mylem}
\emph{Proof}: see Appendix A.\\
It is worth noting that if $\Mc\left(\Am\right)\subset\Mc\left(\Vm\right)$, then the above equality reduces to:
\begin{equation*}
\text{r}\left(\begin{bmatrix}\Am & \Vm \end{bmatrix}\right)=
\text{r}\left(\Vm\right)\Rightarrow \text{r}\left(\Cm^\mathsf{H}\Vm\Cm\right)=
\text{r}\left(\Vm\right)-\text{r}\left(\Am\right)  
\end{equation*} 
which is the same result obtained in \cite{MATRICES_PROD}. 
The second condition of Theorem 1 is satisfied if
\begin{equation}
\text{r}\left(\Pm_l\Delta\widetilde{\Xm}_l\right)\overset{a.s.}
{=}\text{r}\left(\Delta\widetilde{\Xm}_l\right),\
\forall \Delta\widetilde{\Xm}_l\in\Delta\widetilde{\Cc}_l\setminus \left\lbrace \zerov\right\rbrace,\ l=1,\ldots,K.
\label{FD_cond}
\end{equation}
Recalling that $\text{r}\left(\Am\right)=\text{r}\left(\Am\Am^\mathsf{H}\right)$, we have
\begin{equation*} \text{r}\left(\Pm_l\Delta\widetilde{\Xm}_l\right)=
\text{r}\left(\Pm_l\Delta\widetilde{\Xm}_l\Delta\widetilde{\Xm}^\mathsf{H}_l\Pm_l\right).
\end{equation*}
One can easily verify that taking $\Vm=\Delta\widetilde{\Xm}_l\Delta\widetilde{\Xm}^\mathsf{H}_l,\ \Cm=\Pm_l$, and $\Am=\widetilde{\Hcb}_{\overline{l}}$ satisfies \eqref{cond1} and \eqref{cond2}, which implies that
\begin{equation}
\begin{split}
\text{r}\left(\Pm_l\Delta\widetilde{\Xm}_l\right)&=\text{r}\left(\begin{bmatrix}\widetilde{\Hcb}_{\overline{l}}&\Delta\widetilde{\Xm}_l\Delta \widetilde{\Xm}^\mathsf{H}_l\end{bmatrix}\right)-\text{r}\left(\widetilde{\Hcb}_{\overline{l}}\right)\\
&\overset{(a)}{=}\text{r}\left(\begin{bmatrix}\widetilde{\Hcb}_{\overline{l}}&\Delta \widetilde{\Xm}_l\end{bmatrix}\right)-\text{r}\left(\widetilde{\Hcb}_{\overline{l}}\right),
\end{split}
\label{Main_res}
\end{equation}
where $(a)$ follows from $\Mc\left(\Am\right)=\Mc\left(\Am\Am^\mathsf{H}\right)$ for arbitrary matrix $\Am \in \mathbb{C}^{m\times n}$ \cite{MATRIX_ANALYSIS}. Combining Eq.~\eqref{FD_cond}, and Eq.~\eqref{Main_res}, the full-diversity is achieved under PICGD if
\begin{equation*}
\text{r}\left(\begin{bmatrix}\widetilde{\Hcb}_{\overline{l}}&\Delta \widetilde{\Xm}_l\end{bmatrix}\right)\overset{a.s.} {=}\text{r}\left(\widetilde{\Hcb}_{\overline{l}}\right)+\text{r}\left(\Delta\widetilde{\Xm}_l\right),\
\forall\ \Delta\widetilde{\Xm}_l\in\Delta\widetilde{\Cc}_l\setminus \left\lbrace \zerov\right\rbrace,\ l=1,\ldots,K.
\end{equation*}
On the other hand, $\widetilde{\Hcb}_{\overline{l}}$ is of full column rank by definition, and $\Delta\widetilde{\Xm}_l\in\Delta\widetilde{\Cc}_l\setminus \left\lbrace \zerov\right\rbrace$ is of full column rank also thanks to the full-diversity under ML decoding assumption. Thus the above equality is satisfied if and only if $\begin{bmatrix}\widetilde{\Hcb}_{\overline{l}}&\Delta \widetilde{\Xm}_l\end{bmatrix}$ is of full column rank {\bf \emph{almost surely}} or equivalently
\begin{equation*}
\begin{split}
\widetilde{\Hcb}_{\overline{l}}\uv+\Delta\widetilde{\Xm}_l\vv &\overset{a.s.}{\neq} \zerov,\\
\forall\begin{bmatrix}\uv&\vv\end{bmatrix}\in \mathbb{C}\setminus \left\lbrace \zerov\right\rbrace,\forall\ \Delta\widetilde{\Xm}_l\in\Delta\widetilde{\Cc}_l&\setminus \left\lbrace \zerov\right\rbrace,\forall\ l=1,\ldots,K.
\end{split}
\end{equation*}
As the equivalent channel matrix is nothing but vertical concatenation of the equivalent channel matrices at each receive antenna, it suffices to check the above condition for only one receive antenna. A direct consequence of the above discussion is the following lemma.
\begin{mylem}
If STBCs $\Xm_l$ expressed as in \eqref{STBC}
\begin{equation*}
\Xm_l=\sum^{n}_{i=1}\Am_{l,i}s_{l,i},\ s_{k,i}\in\mathbb{C},\ \forall\ l=1,\ldots,K
\end{equation*}
achieve full-diversity under PICGD with grouping scheme ${\sv_1,\ldots,\sv_K}$ in the multi-user MIMO MAC, the rate per user $R$ is upper-bounded as
\begin{equation*}
R\leq\frac{1}{K-1}\left(1-\frac{N_t}{T}\right)
\end{equation*}
\end{mylem}
\begin{proof}
The proof follows directly from the fact that assuming the matrix $\begin{bmatrix}\Hcb_{\overline{l}}&\Delta\Xm_l\end{bmatrix}$ is of full column rank \emph{almost surely} implies that
\begin{equation*}
\left(K-1\right)n+N_t\leq T.
\end{equation*}
By rearranging the terms, and recalling that the rate per user $R$ is equal to $\frac{n}{T}$, the desired result is obtained.
\end{proof}
\begin{mycor}
The STBCs $\Xm_l$ expressed as in \eqref{STBC}
\begin{equation*}
\Xm_l=\sum^{n}_{i=1}\Am_{l,i}s_{l,i},\ s_{l,i}\in\mathbb{C}, \forall\ l=1,\ldots,K
\end{equation*}
achieve full-diversity under PICGD-SIC with grouping scheme ${\sv_1,\ldots,\sv_K}$ in the multi-user MIMO MAC setting if
\begin{itemize}
\item $\text{r}\left(\Pm_l\Delta\widetilde{\Xm}_l\right)\neq 0,\forall\ \Pm_l,\ \Delta\widetilde{\Xm}_l\in\Delta\widetilde{\Cc}_l\setminus \left\lbrace \zerov\right\rbrace,\ l=1,\ldots,K$,
\item The matrix $\Pm_l\Delta\widetilde{\Xm}_l$ is of full column rank {\bf \emph{almost surely}}, $\forall\ \Delta\widetilde{\Xm}_l\in\Delta\widetilde{\Cc}_l\setminus \left\lbrace \zerov\right\rbrace,\ l=1,\ldots,K$,
\end{itemize}
\end{mycor}
\noindent The proof is similar to the PICGD case with the only difference that $\Pm_l\Hcb_{\overline{l}}=\zerov$, where $\widetilde{\Hcb}_{\overline{l}}$ will denote a basis of   $\Mc\left(\begin{bmatrix}\widetilde{\Hcb}_{l+1}&\ldots&\widetilde{\Hcb}_{K}\end{bmatrix}\right)$ and is therefore omitted. 
It can be easily verified that the rate upper-bound in Lemma 2 is still valid for PICGD-SIC. 


\section{Proposed codes structure}
In what follows, we provide two IC schemes for two and three users satisfying the full-diversity conditions in Theorem 1 and Corollary 1, respectively.
 

\subsection{Two-User IC Scheme}
let $\Xm_1$ and $\Xm_2$ be written as
\begin{equation}
\Xm_1=\begin{bmatrix}\Cm\left(\sv_1',N_t\right)\\\underset{1\times N_t}{\zerov}\end{bmatrix},\ \Xm_2=\begin{bmatrix}\underset{1\times N_t}{\zerov}\\\Cm\left(\sv_2',N_t\right)\end{bmatrix}
\label{2_user_IC}
\end{equation}
where 
\begin{equation}
\Cm\left(\sv_i',N_t\right)=\begin{bmatrix}s_{i,1}'&0&\ldots&0&0\\
                                          s_{i,2}'&s_{i,2}'&\ddots&0&0\\
                                          \vdots&\vdots&\ddots&\ddots&\vdots\\
                                          s_{i,N_t-1}'&s_{i,N_t-1}'&\ldots&s_{i,N_t-1}'&0\\
                                          s_{i,N_t}'&s_{i,N_t}'&\ldots&s_{i,N_t}'&s_{i,N_t}'\\
                                          \vdots&\vdots&\ldots&\ldots&\vdots\\
                                          s_{i,n}'&s_{i,n}'&\ldots&s_{i,n}'&s_{i,n}'\\
                                          0&s_{i,1}'&\ldots&s_{i,1}'&s_{i,1}'\\
                                          \vdots&0&\ddots&s_{i,2}'&s_{i,2}'\\
                                          \vdots&\vdots&\ddots&\ddots&\vdots\\
                                          0&0&\ldots&0&s_{i,N_t-1}'\\
                            \end{bmatrix}.
                            \label{C}
\end{equation}
and $\sv_i'=\begin{bmatrix}s_{i,1}'&s_{i,2}'&\ldots&s
_{i,n}'&\end{bmatrix}^\mathsf{T}=\Um_n\begin{bmatrix}s_{i,1}&s_{i,2}&\ldots&s_{i,n}\end{bmatrix}^\mathsf{T}$, $\Um_n$ is the $n\times n$ full-diversity algebraic rotation \cite{FULL_DIV_ROT}, and $s_{i,j},\ \forall\ i=1,2,\ j=1,2,\ldots,n$ are drawn from a conventional QAM constellation $\Ac$. Afterwards, we will refer to this scheme by the two-user IC scheme. For the proposed two-user IC scheme, the signalling period $T$ is equal to $n+N_t$, thus giving rise to a rate per user $R$ of $\frac{n}{n+N_t}$ spcu. Moreover, one can easily verify that the rate per user of the proposed two-user IC scheme achieves the upper-bound established in Lemma 2 with equality. 
\begin{mylem}
The two-user IC scheme in \eqref{2_user_IC} achieves full-diversity under PICGD with grouping scheme $\sv_1$ and $\sv_2$. 
\end{mylem}
\begin{proof}
We consider without loss of generality that $\Xm_1$ is being decoded as the proof for $\Xm_2$ follows similarly.
Towards this end, we prove that $\Xm_1$ achieves full-diversity under PICGD by contradiction. As a preliminary step, we prove by contradiction that $\Xm_1$ achieves the full-diversity under ML decoding. For this purpose, suppose that $\exists\ \vv\neq\zerov, \Delta\Xm_1\in\Delta\Cc_1\setminus\left\lbrace\zerov\right\rbrace\mid \Xm_1\vv=\zerov$, and consider only the first $N_t$ equations
\begin{IEEEeqnarray}{c;c;c}
v_1\Delta s_{1,1}'&=&0\label{PIC_p1_eq_1}\\
\left(v_1+v_2\right)\Delta s_{1,2}'&=&0\label{PIC_p1_eq_2}\\
\vdots&&\vdots\nonumber\\
\left(\sum^{N_t}_{i=1}v_i\right)\Delta s_{1,N_t}'&=&0\label{PIC_p1_eq_Nt}.
\end{IEEEeqnarray}
Recall that the full-diversity algebraic rotations in \cite{FULL_DIV_ROT} are designed to maximize the minimum product distance $d_{p,\min}$ defined as
\begin{equation*}
d_{p,\min}\triangleq\underset{\Delta\sv'=\Um\Delta\sv\mid\Delta\sv\in\mathbb{Z}[i]^n\setminus\left\lbrace\zerov\right\rbrace}{\min}\left\lbrace\prod^{n}_{i=1}\vert \Delta s_i'\vert\right\rbrace.
\end{equation*}  
Restricting the rotation matrices to those provided in \cite{FULL_DIV_ROT} implies that $\Delta s_i'\neq 0\ \forall\ \Delta \sv\in\mathbb{Z}[i]^n\setminus\left\lbrace\zerov\right\rbrace, i=1,\ldots,n$. Accordingly, from \eqref{PIC_p1_eq_1} we have $v_1=0$, otherwise we will have $\Delta\sv_1=\zerov$ or equivalently $\Delta\Xm_1=\zerov$ thanks to the aforementioned properties of full-diversity algebraic rotation matrices. Consequently, \eqref{PIC_p1_eq_2} implies that $v_2=0$. Proceeding in this manner, we get $\vv=\zerov$ which contradicts our assumption and concludes the first part of the proof. It can be easily verified that $\Hcb_{\overline{1}}=\Hcb_{2,1}=\begin{bmatrix}\underset{n\times 1}{\zerov}&\Dm\left(\hv_{2,1}\right)^\mathsf{T}\end{bmatrix}^\mathsf{T}$, where for a vector $\hv=\begin{bmatrix}h_1,\ldots,h_{N_t}\end{bmatrix}^\mathsf{T}$, one has
\begin{equation}
\underset{\left(n+N_t-1\right)\times n}{\Dm\left(\hv\right)}=
\begin{bmatrix}h_1&0&\ldots&0&0&\ldots&0&0\\
                0&h_1+h_2&\ddots&0&0&\ldots&0&0\\
                \vdots&\ddots&\ddots&\ddots&\vdots&\ldots&\vdots&\vdots\\ 
                0&0&\ddots&\sum^{N_t-1}_{i=1}h_i&0&\ldots&0&0\\
                0&0&\ldots&0&\sum^{N_t}_{i=1}h_i&\ddots&0&0\\
                \vdots&\vdots&\ldots&\vdots&\ddots&\ddots&\ddots&\vdots\\
                0&0&\ldots&0&0&\ddots&\sum^{N_t}_{i=1}h_i&0\\
                0&0&\ldots&0&0&\ldots&0&\sum^{N_t}_{i=1}h_i\\
                \sum^{N_t}_{i=2}h_i&0&\ldots&0&0&\ldots&0&0\\
                0&\sum^{N_t}_{i=3}h_i&\ddots&0&0&\ldots&0&0\\
                \vdots&\ddots&\ddots&\ddots&\vdots&\ldots&\vdots&\vdots\\
                0&0&&h_{N_t}&0&\ldots&0&0                 
\end{bmatrix}.
\label{D}
\end{equation}
We proceed now towards the main body of our proof and demonstrate that the proposed two-user IC scheme satisfies Theorem 1 by contradiction as well.
Suppose that $\exists\ \Delta\Xm_1\in\Delta\Cc_1\setminus\left\lbrace\zerov\right\rbrace$ which lies completely in the subspace spanned by the columns of $\Hcb_{\overline{1}}$. The above cannot be true unless $\Delta s_{1,1}'=0$, which in turns implies that $\Delta\Xm_1=\zerov$ thanks to the full-diversity algebraic rotation matrices and hence contradicts our assumption. Now suppose that $\exists \begin{bmatrix}\uv&\vv\end{bmatrix}\neq\zerov, \Delta\Xm_1\in\Cc_1\setminus\left\lbrace\zerov\right\rbrace\mid
\Hcb_{\overline{1}}\uv+\Delta\Xm_1\vv=\zerov$ with {\bf \emph{non-zero probability}}.
Therefore, one has the following set of equations
\begin{IEEEeqnarray}{c;c;c}
v_1\Delta s_{1,1}'&=&0\label{PIC_p2_eq_1}\\
\left(v_1+v_2\right)\Delta s_{1,2}'+\hv_{2,1}\left(1\right)u_1&=&0\label{PIC_p2_eq_2}\\
\left(v_1+v_2+v_3\right)\Delta s_{1,3}'+\left(\hv_{2,1}\left(1\right)+\hv_{2,1}\left(2\right)\right)u_2&=&0\label{PIC_p2_eq_3}\\
\vdots&&\vdots\nonumber\\
\left(\sum^{N_t}_{i=1}v_i\right)\Delta s_{1,N_t}'+\left(\sum^{N_t-1}_{j=1}\hv_{2,1}\left(j\right)\right)u_{N_t-1}&=&0\label{PIC_p2_eq_Nt}\\
\left(\sum^{N_t}_{i=1}v_i\right)\Delta s_{1,N_t+1}'+\left(\sum^{N_t}_{j=1}\hv_{2,1}\left(j\right)\right)u_{N_t}&=&0\label{PIC_p2_eq_Nt+1}\\
\vdots&&\vdots\nonumber\\
\left(\sum^{N_t}_{i=1}v_i\right)\Delta s_{1,n}'+\left(\sum^{N_t}_{j=1}\hv_{2,1}\left(j\right)\right)u_{n-1}&=&0\label{PIC_p2_eq_n}\\
\left(\sum^{N_t}_{i=2}v_i\right)\Delta s_{1,1}'+\left(\sum^{N_t}_{j=1}\hv_{2,1}\left(j\right)\right)u_n&=&0\label{PIC_p2_eq_n+1}\\
\left(\sum^{N_t}_{i=3}v_i\right)\Delta s_{1,2}'+\left(\sum^{N_t}_{j=2}\hv_{2,1}\left(j\right)\right)u_1&=&0\label{PIC_p2_eq_n+2}\\
\left(\sum^{N_t}_{i=4}v_i\right)\Delta s_{1,3}'+\left(\sum^{N_t}_{j=3}\hv_{2,1}\left(j\right)\right)u_2&=&0\label{PIC_p2_eq_n+3}\\
\vdots&&\vdots\nonumber\\
v_{N_t}\Delta s_{1,N_t-1}'+\left(\sum^{N_t}_{j=N_t-1}\hv_{2,1}\left(j\right)\right)u_{N_t-2}&=&0\\
\hv_{2,1}\left(N_t\right)u_{N_t-1}&=&0\label{PIC_p2_eq_T}
\end{IEEEeqnarray}
from \eqref{PIC_p2_eq_1} we have $v_1=0$, otherwise taking $\Delta s_{1,1}'=0$ implies that $\Delta \sv_1=\zerov$ or equivalently $\Delta\Xm_1=\zerov$ thanks to the full-diversity algebraic rotation. On the other hand, \eqref{PIC_p2_eq_T} implies that $u_{Nt-1}=0$ as $\hv_{2,1}\left(N_t\right)=0$ is a {\bf \emph{zero-probability}} event and is thus discarded since isolated events do not violate the second condition in Theorem 1. Consequently, one has from \eqref{PIC_p2_eq_Nt} that $\sum^{N_t}_{i=1}v_i=0$. Recall that the entries of $\hv_{2,1}$ are i.i.d., therefore
\begin{equation*} 
\mathbb{P}\left[\underset{i\in\Ic}{\sum}\hv_{2,1}\left(i\right)=0\right]=0,\ \forall\ \Ic\subseteq\left\lbrace 1,\ldots,N_t\right\rbrace.
\end{equation*}
Hence, thanks to \eqref{PIC_p2_eq_Nt+1}, \eqref{PIC_p2_eq_n} and \eqref{PIC_p2_eq_n+1}, one has that $\begin{bmatrix}u_{N_t}&\ldots&u_n\end{bmatrix}=\zerov$. Adding \eqref{PIC_p2_eq_n+2} and \eqref{PIC_p2_eq_2} results in $u_1=0$. Recalling that $v_1=0$, yields $v_2=0$ thanks to \eqref{PIC_p2_eq_2}. Similarly, adding \eqref{PIC_p2_eq_n+3} and \eqref{PIC_p2_eq_3} yields $u_2=0$ and $v_3=0$ thanks to \eqref{PIC_p2_eq_3}. Proceeding in the same manner, results in $\begin{bmatrix}\uv&\vv\end{bmatrix}=\zerov$ which contradicts our initial assumption and thus completes the proof.
\end{proof}
\begin{myexm}
Consider the following rate-3/5 two-user IC scheme:
\begin{equation}
\Xm_1=\begin{bmatrix}s_{1,1}'&s_{1,2}'&s_{1,3}'&0&0\\0&s_{1,2}'&s_{1,3}'&s_{1,1}'&0\end{bmatrix}^\mathsf{T}, \Xm_2=\begin{bmatrix}0&s_{2,1}'&s_{2,2}'&s_{2,3}'&0\\0&0&s_{2,2}'&s_{2,3}'&s_{2,1}'\end{bmatrix}^\mathsf{T},
\label{ex1_LSTBC}
\end{equation}
\end{myexm}
\noindent where $\begin{bmatrix}s_{i,1}'&s_{i,2}'&s
_{i,3}'&\end{bmatrix}^\mathsf{T}=\Um_3\begin{bmatrix}s_{i,1}&s_{i,2}&s_{i,3}\end{bmatrix}^\mathsf{T}$,$\Um_3$ is the $3\times 3$ full diversity rotation \cite{FULL_DIV_ROT}  and $s_{i,j},\ \forall\ i=1,2,\ j=1,2,3$ are drawn from a conventional QAM constellation $\Ac$. 
\begin{mylem}
For the two-user IC scheme in \eqref{2_user_IC}, the employment of the real full-diversity algebraic rotations enables separate decoding of the real and imaginary parts of $\sv_1$ and $\sv_2$ under PICGD without any loss of performance.
\end{mylem}
\emph{Proof}: see Appendix B.


\subsection{Three-User IC Scheme}
let $\sv_i'=\begin{bmatrix}s_{i,1}'&s_{i,2}'&\ldots&s
_{i,n}'&\end{bmatrix}^\mathsf{T}=\Um_{n}\begin{bmatrix}s_{i,1}&s_{i,2}&\ldots&s_{i,n}\end{bmatrix}^\mathsf{T}$, therefore our three-user IC scheme can be expressed as
\begin{equation}
\Xm_1=\begin{bmatrix}\Cm\left(\sv_1',N_t\right)\\
\underset{\left(n+1\right)\times N_t}{\zerov}\end{bmatrix},
\Xm_2=\begin{bmatrix}\underset{N_t\times N_t}{\zerov}\\\Cm\left(\sv_2',N_t\right)\\
\underset{\left(n-N_t+1\right)\times N_t}{\zerov}\end{bmatrix},
\Xm_3=\begin{bmatrix}\underset{\left(n+1\right)\times N_t}{\zerov}\\\Cm\left(\sv_3',N_t\right)\end{bmatrix},
\label{3_user_IC}
\end{equation}
where $n\geq 2N_t-1$, and $\Cm\left(\sv_i,N_t\right)$ is as defined in \eqref{C}. Afterwards, we will refer to this scheme by the three-user IC scheme. For the proposed three-user IC scheme, the signalling period $T$ is equal to $2n+N_T$, hence resulting in a rate per user $R$ of $\frac{n}{2n+N_t}$ spcu. Furthermore, it is straightforward to verify that the rate per user of the provided three-user IC scheme satisfies the upper-bound in Lemma 2 with equality.
\begin{mylem}
The three-user IC scheme in \eqref{3_user_IC}, achieves full-diversity under PICGD-SIC with ordered grouping scheme $\sv_1$, $\sv_2$ and $\sv_3$. 
\end{mylem}
\begin{proof}
It can be easily checked that $\Xm_i,i=1,2,3$ achieve full-diversity under ML decoding in a similar fashion to the two-user IC case.  It can be easily verified that $\Hcb_{\overline{1}}=\begin{bmatrix}\Hcb_{2,1}&\Hcb_{3,1}\end{bmatrix}$. Now, we proceed  to the proof that the scheme in \eqref{3_user_IC} satisfies Corollary 1. 
Suppose that $\exists\ \Delta\Xm_1\in\Delta\Cc_1\setminus\left\lbrace\zerov\right\rbrace$ which lies completely in the subspace spanned by $\Hcb_{\overline{1}}$. 
This cannot be true unless $\Delta s_{1,1}'=\Delta s_{1,2}'=\ldots=\Delta s_{1,N_t}'=0$, or equivalently $\Delta\Xm_1=\zerov$ which contradicts our assumption. 
Next, suppose that $\exists \begin{bmatrix}\uv&\vv\end{bmatrix}\neq\zerov, \Delta\Xm_1\in\Delta\Cc_1\setminus{\lbrace\zerov\rbrace} \mid \Hcb_{\overline{1}}\uv+\Delta\Xm_1\vv=\zerov$ with {\bf \emph{non-zero probability}}. According to the IC scheme in \eqref{3_user_IC}, the first $N_t$ equations are
\begin{IEEEeqnarray}{c;c;c}
v_1\Delta s_{1,1}'&=&0\\
\left(v_1+v_2\right)\Delta s_{1,2}'&=&0\\
\vdots&&\vdots\nonumber\\
\left(\sum^{N_t}_{i=1}v_i\right)\Delta s_{1,N_t}'&=&0.
\end{IEEEeqnarray}
Thanks to the full-diversity algebraic rotations, the above system of linear equations implies that $v_1=v_2=\ldots\ v_{N_t}=0$. The rest of the equations are
\begin{IEEEeqnarray}{c;c;c}
\hv_{2,1}\left(1\right)u_1&=&0\\
\vdots&&\vdots\nonumber\\
\left(\sum^{N_t}_{i=1}\hv_{2,1}\left(i\right)\right)u_{n-N_t+1}&=&0\\
\left(\sum^{N_t}_{i=1}\hv_{2,1}\left(i\right)\right)u_{n-N_t+2}+\hv_{3,1}\left(1\right)u_{n+1}&=&0\label{PIC_SIC_p2_eq_n+2}\\
\vdots&&\vdots\nonumber\\
\left(\sum^{N_t}_{i=1}\hv_{2,1}\left(i\right)\right)u_n+\left(\sum^{N_t-1}_{j=1}\hv_{3,1}\left(j\right)\right)u_{n+N_t-1}&=&0\\
\left(\sum^{N_t}_{i=2}\hv_{2,1}\left(i\right)\right)u_1+\left(\sum^{N_t}_{j=1}\hv_{3,1}\left(j\right)\right)u_{n+N_t}&=&0\label{PIC_SIC_p2_eq_n+N_t+1}\\
\vdots&&\vdots\nonumber\\
\hv_{2,1}\left(N_t\right)u_{N_t-1}+\left(\sum^{N_t}_{j=1}\hv_{3,1}\left(j\right)\right)u_{n+2N_t-2}&=&0\\
\left(\sum^{N_t}_{j=1}\hv_{3,1}\left(j\right)\right)u_{n+2N_t-1}&=&0\\
\vdots&&\vdots\nonumber\\
\left(\sum^{N_t}_{j=1}\hv_{3,1}\left(j\right)\right)u_{2n}&=&0\\
\left(\sum^{N_t}_{j=2}\hv_{3,1}\left(j\right)\right)u_{n+1}&=&0\\
\vdots&&\vdots\nonumber\\
\hv_{3,1}\left(N_t\right)u_{n+N_t-1}&=&0.
\end{IEEEeqnarray}
From the first $n-N_t+1$ equations, one has $u_1=u_2=\ldots =u_{n-N_t+1}=0$. Similarly, from the last $n-N_t+1$ equations, one has $u_{n+2N_t-1}=u_{n+2N_t}=\ldots=u_{2n}=0$ and $u_{n+1}=u_{n+2}=\ldots=u_{n+N_t-1}=0$. Consequently,  \eqref{PIC_SIC_p2_eq_n+2} implies that $u_{n-N_t+2}=0$. Proceeding in the same manner, we obtain $u_{n-N_t+2}=u_{n-N_t+3}=\ldots =u_{n}=0$. On the other hand, thanks to \eqref{PIC_SIC_p2_eq_n+N_t+1}, one gets $u_{n+N_t}=0$. Proceeding similarly we get $u_{n+N_t}=u_{n+N_t=1}=\ldots=u_{n+2N_t-2}=0$. Therefore one has $\begin{bmatrix}\uv&\vv\end{bmatrix}=\zerov$ which contradicts our initial assumption. 
Next, $\Xm_2$ is decoded, thus suppose that $\exists\ \Delta\Xm_2\in\Delta\Cc_2\setminus\left\lbrace\zerov\right\rbrace$ which lies completely in the subspace spanned by the columns of $\Hcb_{\overline{2}}$. This cannot be true unless $\Delta s_{2,1}'=\Delta s_{2,2}'=\ldots=\Delta s_{2,n-N_t+1}'=0$, hence $\Delta\Xm_2=\zerov$ which contradicts the assumption. Next, $\exists \begin{bmatrix}\uv&\vv\end{bmatrix}\neq\zerov, \Delta\Xm_2\in\Delta\Cc_2\setminus{\lbrace\zerov\rbrace} \mid \Hcb_{\overline{2}}\uv+\Delta\Xm_2\vv=\zerov$ with {\bf \emph{non-zero probability}}. It can be easily verified that in this case $\Hcb_{\overline{2}}=\Hcb_{3,1}$, thus we have
\begin{IEEEeqnarray}{c;c;c}
v_1\Delta s_{2,1}'&=&0\\
\vdots&&\vdots\nonumber\\
\left(\sum^{N_t}_{i=1}v_i\right)\Delta s_{2,n-N_t+1}'&=&0\\
\left(\sum^{N_t}_{i=1}v_i\right)\Delta s_{2,n-N_t+2}'+\hv_{3,1}\left(1\right)u_1&=&0\label{PIC_SIC_p3_eq_n-N_t+2}\\
\vdots&&\vdots\nonumber\\
\left(\sum^{N_t}_{i=1}v_i\right)\Delta s_{2,n}'+\left(\sum^{N_t-1}_{j=1}\hv_{3,1}\left(j\right)\right)u_{N_t-1}&=&0\label{PIC_SIC_p3_eq_n}\\
\left(\sum^{N_t}_{i=2}v_i\right)\Delta s_{2,1}'+\left(\sum^{N_t}_{j=1}\hv_{3,1}\left(j\right)\right)u_{N_t}&=&0\label{PIC_SIC_p3_eq_n+1}\\
\vdots&&\vdots\nonumber\\
v_{N_t}\Delta s_{2,N_t-1}'+\left(\sum^{N_t}_{j=1}\hv_{3,1}\left(j\right)\right)u_{2N_t-2}&=&0\label{PIC_SIC_p3_eq_2n-N_t}\\
\left(\sum^{N_t}_{j=1}\hv_{3,1}\left(j\right)\right)u_{2N_t-1}&=&0\\
\vdots&&\vdots\nonumber\\
\left(\sum^{N_t}_{j=1}\hv_{3,1}\left(j\right)\right)u_n&=&0\\
\left(\sum^{N_t}_{j=2}\hv_{3,1}\left(j\right)\right)u_1&=&0\\
\vdots&&\vdots\nonumber\\
\hv_{3,1}\left(N_t\right)u_{N_t-1}&=&0.
\end{IEEEeqnarray}
From the first $N_t$ equations one obtains $v_1=v_2=\ldots =v_{N_t}=0$ thanks to the full-diversity algebraic rotation. Consequently, according to the rest of equation we get $u_1=u_2=\ldots =u_n=0$, which contradicts our initial assumption. The proof for $\Xm_3$ is trivial since it achieves the full-diversity under ML, which finalizes the proof.
\end{proof}
\begin{myexm}
Consider the following rate-3/8 three-user IC scheme:
\begin{IEEEeqnarray}{c;c;c}
\Xm_1&=&\begin{bmatrix}s_{1,1}'&s_{1,2}'&s_{1,3}'&0&0&0&0&0\\
0&s_{1,2}'&s_{1,3}'&s_{1,1}'&0&0&0&0\end{bmatrix}^\mathsf{T},\nonumber\\
\Xm_2&=&\begin{bmatrix}0&0&s_{2,1}'&s_{2,2}'&s_{2,3}'&0&0&0\\
0&0&0&s_{2,2}'&s_{2,3}'&s_{2,1}'&0&0\end{bmatrix}^\mathsf{T},\\
\Xm_3&=&\begin{bmatrix}0&0&0&0&s_{3,1}'&s_{3,2}'&s_{3,3}'&0\\
0&0&0&0&0&s_{3,2}'&s_{3,3}'&s_{3,1}'\end{bmatrix}^\mathsf{T}\nonumber
\end{IEEEeqnarray}
\end{myexm}
\noindent where $\begin{bmatrix}s_{i,1}'&s_{i,2}'&s
_{i,3}'&\end{bmatrix}^\mathsf{T}=\Um_3\begin{bmatrix}s_{i,1}&s_{i,2}&s_{i,3}\end{bmatrix}^\mathsf{T}$,$\Um_3$ is the $3\times 3$ full diversity rotation \cite{FULL_DIV_ROT}  and $s_{i,j},\ \forall\ i=1,2,3,\ j=1,2,3$ are drawn from a conventional QAM constellation $\Ac$. 
\begin{mylem}
For the three-user IC scheme in \eqref{3_user_IC}, the employment of the real full-diversity algebraic rotations enables separate decoding of the real and imaginary parts of $\sv_1,\sv_2$ and $\sv_3$ under PICGD-SIC without any loss of performance.
\end{mylem}
\emph{Proof}: see Appendix C.

A natural question arises here about the existence of similar schemes for arbitrary number of users. In fact, one could obtain similar schemes for an arbitrary number of users at the expense of a decreasing rate per user $R$. It can be proven that for $K$ users, $R$ is equal to $\frac{n}{\left(K-1\right)n+N_t}$ spcu, where $n$ denotes the number of transmitted symbols per codeword. For the case of two users (i.e., $K=2$), $R$ becomes $\frac{n}{n+N_t}$ which approaches $1$ asymptotically. For the case of three users (i.e., $K=3$), $R$ becomes $\frac{n}{2n+N_t}$, which approaches $\frac{1}{2}$ asymptotically. In the case of $K$ users, the asymptotic rate per user is equal to $\frac{1}{K-1}$, which approaches the rate of TDMA Alamouti \cite{ALAMOUTI} signalling (i.e., $\frac{1}{K}$) for large $K$. In other words, increasing the number of users decreases the rate per user to the extent that TDMA may become an attractive alternative. 


\section{Simulations results}
In this section, we corroborate our theoretical claims via numerical simulations. In the first part, the proposed two-user IC scheme is shown to achieve the full-diversity gain offered by the MIMO MAC configuration. For this purpose, the CER performance of the proposed two-user IC schemes is compared to the reference diversity slopes in two MIMO MAC configurations, namely $\left(2^2,1\right)$ and $\left(3^2,1\right)$. The CER performance curves for our rate-$1/2$, $3/5$, and $2/3$ two-user IC schemes in the MIMO MAC $\left(2^2,1\right)$ configuration are depicted in Fig.~\ref{2x2x1_ref}, while the CER performance curves for our rate-$1/2$ and $4/7$ two-user IC schemes in the MIMO MAC $\left(3^2,1\right)$ configuration are depicted in Fig.~\ref{3x3x1_ref}. As can be easily verified, our proposed two-user IC scheme achieves the full-diversity gain as predicted by Lemma 3. 
%
%
%
%
\begin{figure}[h!]
\centering 
\begin{tikzpicture}

\begin{semilogyaxis}[%
view={0}{90},
width=4.0185in,
height=3.1694in,
scale only axis,
xmin=0, xmax=50,
xlabel={$\frac{E_b}{N_0}$ in dB},
xmajorgrids,
ymin=1e-06, ymax=1,
yminorticks=true,
ylabel={CER},
ymajorgrids,
yminorgrids,
legend style={at={(0.02,0.02)},anchor=south west,legend cell align=left}]
\addplot [
color=blue,
solid,
line width=2.0pt,
mark=o,
mark options={solid}
]
coordinates{
 (0,0.578846666666667)(5,0.324465555555556)(10,0.121684444444444)(15,0.03043)(20,0.00540777777777778)(25,0.000752222222222222)(30,8.55555555555556e-05)(35,8.88888888888889e-06) 
};
\addlegendentry{\footnotesize Rate-1/2 (i.e., n=2) IC scheme, QPSK};

\addplot [
color=blue,
solid,
line width=2.0pt,
mark=square,
mark options={solid}
]
coordinates{
 (0,0.81982)(5,0.612341666666666)(10,0.336541666666667)(15,0.126623333333333)(20,0.032115)(25,0.00598833333333334)(30,0.000881666666666667)(35,0.000116666666666667)(40,1.5e-05)(45,1.66666666666667e-06) 
};
\addlegendentry{\footnotesize Rate-3/5 (i.e., n=3) IC scheme, QPSK};

\addplot [
color=blue,
solid,
line width=2.0pt,
mark=triangle,
mark options={solid,rotate=270}
]
coordinates{
 (0,0.936797037037037)(5,0.825082962962963)(10,0.598134814814815)(15,0.327474074074074)(20,0.132905185185185)(25,0.0401037037037037)(30,0.00911111111111111)(35,0.00162518518518518)(40,0.000248888888888889)(45,3.55555555555556e-05)(50,4.44444444444444e-06) 
};
\addlegendentry{\footnotesize Rate-2/3 (i.e., n=4) Ic scheme, QPSK};

\addplot [
color=blue,
dashed,
line width=2.0pt
]
coordinates{
 (20,0.0169)(25,0.00169)(30,0.000169)(35,1.69e-05) 
};
\addlegendentry{\footnotesize Diversity gain $2$ reference};

\addplot [
color=blue,
dashed,
line width=2.0pt,
forget plot
]
coordinates{
 (30,0.003025)(35,0.0003025)(40,3.025e-05)(45,3.025e-06) 
};
\addplot [
color=blue,
dashed,
line width=2.0pt,
forget plot
]
coordinates{
 (35,0.0075625)(40,0.00075625)(45,7.5625e-05)(50,7.5625e-06) 
};
\end{semilogyaxis}
\end{tikzpicture}
\caption{Codeword error rate performance for the $\left(2^2,1\right)$ MIMO MAC channel.}
\label{2x2x1_ref}
\end{figure}
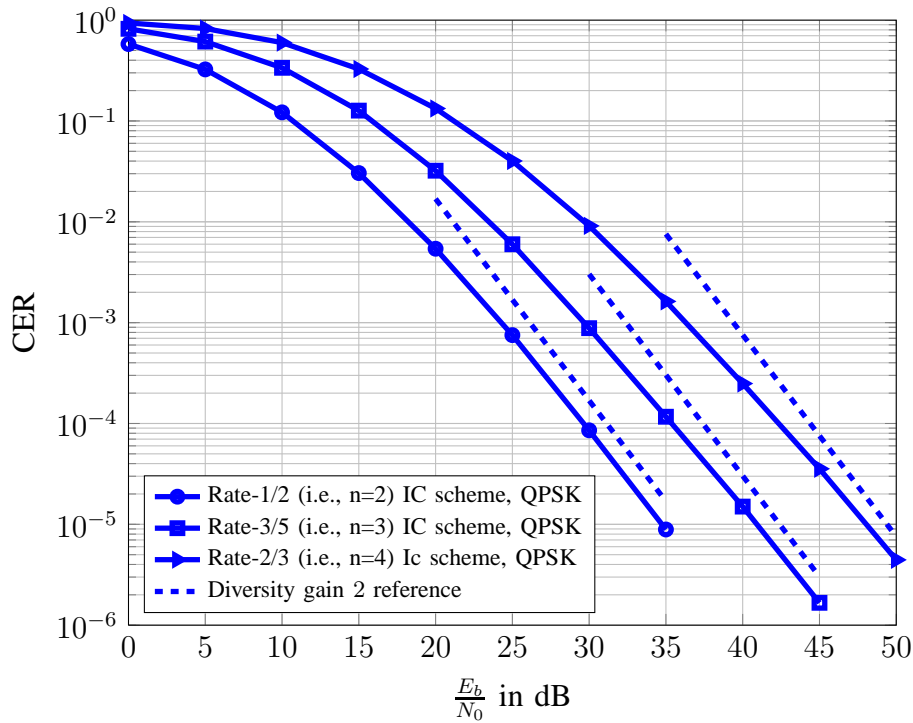
%
%
%
%
\begin{figure}[h!]
\centering 
\begin{tikzpicture}

\begin{semilogyaxis}[%
view={0}{90},
width=4.0185in,
height=3.1694in,
scale only axis,
xmin=0, xmax=45,
xlabel={$\frac{E_b}{N_0}$ in dB},
xmajorgrids,
ymin=1e-07, ymax=1,
yminorticks=true,
ylabel={CER},
ymajorgrids,
yminorgrids,
legend style={at={(0.02,0.02)},anchor=south west,legend cell align=left}]
\addplot [
color=blue,
solid,
line width=2.0pt,
mark=o,
mark options={solid}
]
coordinates{
 (0,0.742218333333333)(3.5,0.552091666666666)(7,0.330895833333333)(10.5,0.15149875)(14,0.05163125)(17.5,0.0132016666666667)(21,0.00262833333333333)(24.5,0.000426666666666667)(28,5.58333333333334e-05)(31.5,5.83333333333333e-06) 
};
\addlegendentry{\footnotesize Rate-1/2 (i.e., n=3) IC scheme, QPSK};

\addplot [
color=blue,
solid,
line width=2.0pt,
mark=square,
mark options={solid}
]
coordinates{
 (0,0.889528666666667)(3.5,0.763827555555554)(7,0.563248444444444)(10.5,0.335616666666667)(14,0.155802)(17.5,0.0564313333333333)(21,0.01649)(24.5,0.00407466666666666)(28,0.000886133333333323)(31.5,0.000174000000000007)(35,3.13333333333335e-05)(38.5,4.84444444444444e-06)(42,6.66666666666667e-07) 
};
\addlegendentry{\footnotesize Rate-4/7 (i.e., n=4) IC scheme, QPSK};

\addplot [
color=blue,
dashed,
line width=2.0pt
]
coordinates{
 (21,0.0164229032714985)(24.5,0.00146369279476027)(28,0.00013045175764697)(31.5,1.16265251384058e-05) 
};
\addlegendentry{\footnotesize Diversity gain $3$ reference};

\addplot [
color=blue,
dashed,
line width=2.0pt,
forget plot
]
coordinates{
 (31.5,0.00190157800792369)(35,0.000169478318349649)(38.5,1.51047710222454e-05)(42,1.34621413438716e-06) 
};
\end{semilogyaxis}
\end{tikzpicture}
\caption{Codeword error rate performance for the $\left(3^2,1\right)$ MIMO MAC channel.}
\label{3x3x1_ref}
\end{figure}
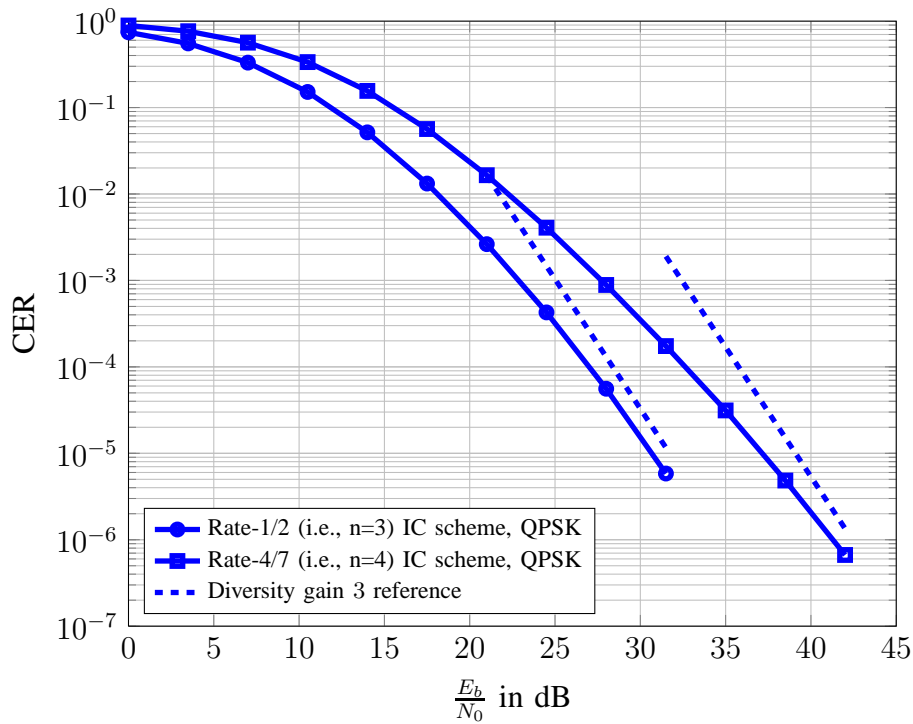

In the second part, the full diversity claim of the proposed three-user IC scheme is verified through numerical simulations. Towards this end we compare the CER performance of the proposed three-user IC schemes to the reference diversity slopes in two MIMO MAC configurations, namely $\left(2^3,1\right)$ and $\left(3^3,1\right)$. The CER performance curves for our rate-$3/8$ and $2/5$ three-user IC schemes in the MIMO MAC $\left(2^3,1\right)$ configuration are depicted in Fig.~\ref{2x2x2x1_ref}, while the CER performance curves for our rate-$5/13$ and $2/5$ three-user IC schemes in the MIMO MAC $\left(3^3,1\right)$ configuration are depicted in Fig.~\ref{3x3x3x1_ref}. Clearly, our proposed three-user IC scheme achieves the full-diversity gain as predicted by Lemma 5. 
%
%
%
%
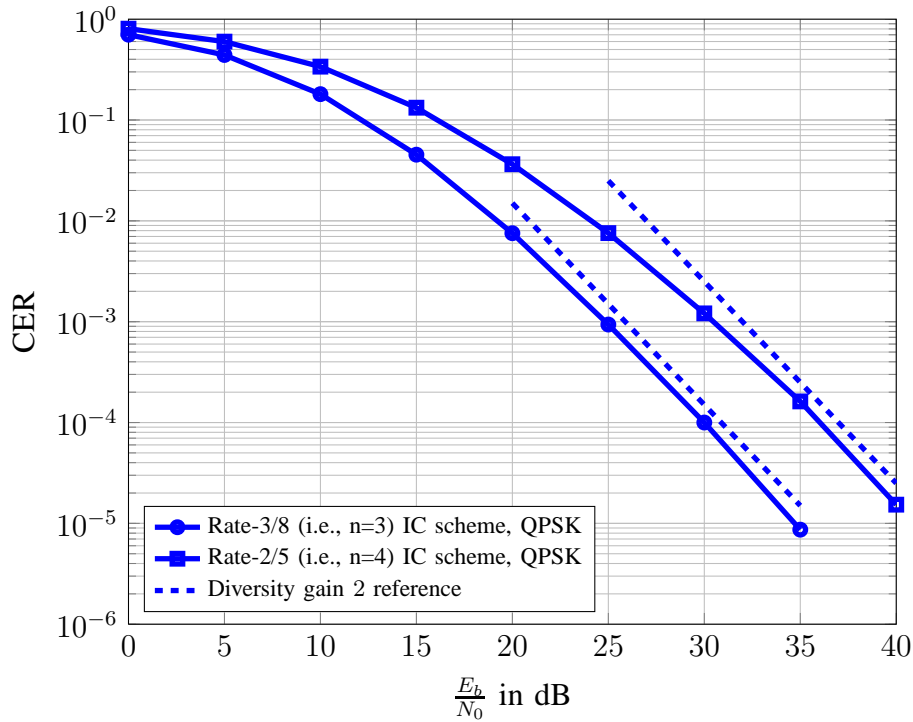
\begin{figure}[h!]
\centering 
\begin{tikzpicture}

\begin{semilogyaxis}[%
view={0}{90},
width=4.0185in,
height=3.1694in,
scale only axis,
xmin=0, xmax=40,
xlabel={$\frac{E_b}{N_0}$ in dB},
xmajorgrids,
ymin=1e-06, ymax=1,
yminorticks=true,
ylabel={CER},
ymajorgrids,
yminorgrids,
legend style={at={(0.02,0.02)},anchor=south west,legend cell align=left}]
\addplot [
color=blue,
solid,
line width=2.0pt,
mark=o,
mark options={solid}
]
coordinates{
 (0,0.701639333333333)(5,0.44146)(10,0.180454)(15,0.0453026666666667)(20,0.00753866666666666)(25,0.000937000000000001)(30,0.0001)(35,8.66666666666666e-06) 
};
\addlegendentry{\footnotesize Rate-3/8 (i.e., n=3) IC scheme, QPSK};

\addplot [
color=blue,
solid,
line width=2.0pt,
mark=square,
mark options={solid}
]
coordinates{
 (0,0.804613666666665)(5,0.598656666666668)(10,0.338151666666667)(15,0.132428333333334)(20,0.036438)(25,0.00756133333333332)(30,0.001209)(35,0.000162)(40,1.53333333333333e-05) 
};
\addlegendentry{\footnotesize Rate-2/5 (i.e., n=4) IC scheme, QPSK};

\addplot [
color=blue,
dashed,
line width=2.0pt
]
coordinates{
 (25,0.025)(30,0.0025)(35,0.00025)(40,2.5e-05) 
};
\addlegendentry{\footnotesize Diversity gain $2$ reference};

\addplot [
color=blue,
dashed,
line width=2.0pt,
forget plot
]
coordinates{
 (20,0.015)(25,0.0015)(30,0.00015)(35,1.5e-05) 
};
\end{semilogyaxis}
\end{tikzpicture}
\caption{Codeword error rate performance for the $\left(2^3,1\right)$ MIMO MAC channel.}
\label{2x2x2x1_ref}
\end{figure}
%
%
%
%
\begin{figure}[h!]
\centering 
\begin{tikzpicture}

\begin{semilogyaxis}[%
view={0}{90},
width=4.0185in,
height=3.1694in,
scale only axis,
xmin=0, xmax=40,
xlabel={$\frac{E_b}{N_0}$ in dB},
xmajorgrids,
ymin=1e-06, ymax=1,
yminorticks=true,
ylabel={CER},
ymajorgrids,
yminorgrids,
legend style={at={(0.02,0.02)},anchor=south west,legend cell align=left}]
\addplot [
color=blue,
solid,
line width=2.0pt,
mark=o,
mark options={solid}
]
coordinates{
 (0,0.863916066666666)(3.5,0.729326266666665)(7,0.525610200000001)(10.5,0.2987002)(14,0.128444333333333)(17.5,0.0414774666666667)(21,0.0100263333333333)(24.5,0.00185253333333333)(28,0.000270266666666658)(31.5,2.93333333333334e-05)(35,2.53333333333333e-06) 
};
\addlegendentry{\footnotesize Rate-5/13 (i.e., n=5) IC scheme, QPSK};

\addplot [
color=blue,
solid,
line width=2.0pt,
mark=square,
mark options={solid}
]
coordinates{
 (0,0.896963333333333)(3.5,0.790463333333333)(7,0.631066666666667)(10.5,0.427463333333333)(14,0.231206666666667)(17.5,0.0973066666666666)(21,0.0313)(24.5,0.00764)(28,0.00139333333333333)(31.5,0.000213133333333322)(35,2.69666666666666e-05)(38.5,3.53333333333334e-06) 
};
\addlegendentry{\footnotesize Rate-2/5 (i.e., n=6) IC scheme, QPSK};

\addplot [
color=blue,
dashed,
line width=2.0pt
]
coordinates{
 (28,0.00796214341106995)(31.5,0.000709626778467151)(35,6.32455532033676e-05)(38.5,5.6367658625289e-06) 
};
\addlegendentry{\footnotesize Diversity gain $3$ reference};

\addplot [
color=blue,
dashed,
line width=2.0pt,
forget plot
]
coordinates{
 (24.5,0.00670025388226444)(28,0.000597160755830247)(31.5,5.32220083850364e-05)(35,4.74341649025257e-06) 
};
\end{semilogyaxis}
\end{tikzpicture}
\caption{Codeword error rate performance for the $\left(3^3,1\right)$ MIMO MAC channel.}
\label{3x3x3x1_ref}
\end{figure}
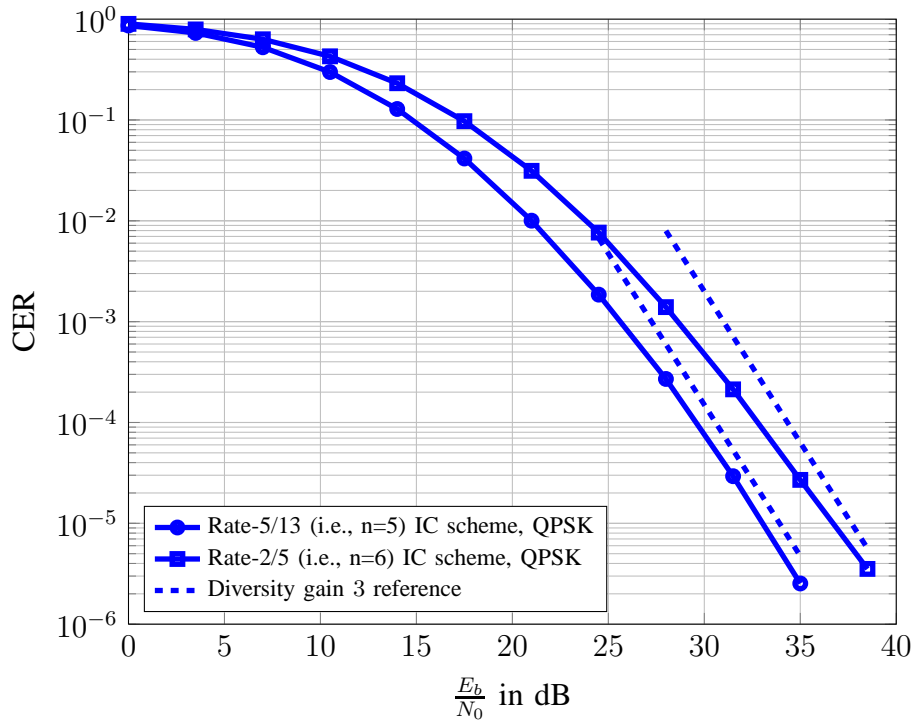

Finally, we compare our two-user IC scheme in Eq.~\eqref{2_user_IC} to the interference cancellation scheme in \cite{2-USER_MIMO_MAC_PIC}. The performance is provided in terms of the CER over Rayleigh fading channel and the average decoding complexity for the $\left(2^2,4\right)$ MIMO MAC configuration. The ML detection is performed via a depth-first tree traversal with infinite initial radius sphere decoder. The radius is updated whenever a leaf node is reached and sibling nodes are visited according to the Schnorr-Euchner enumeration \cite{SE}. The proposed IC schemes and the one proposed in \cite{2-USER_MIMO_MAC_PIC} achieve the full-diversity gain under PICGD. However the latter scheme suffers from a rate loss w.r.t. our two-user IC scheme for the same signalling period. In fact, the rate per user for our two-user IC scheme is $\frac{n}{n+N_t}$, whereas for the rate per user for the IC scheme in \cite{2-USER_MIMO_MAC_PIC} is $\frac{n'}{n'+2N_t-1}$, where $n$ (resp. $n'$) denotes the number of codeword symbols for our IC scheme (resp. the IC scheme in \cite{2-USER_MIMO_MAC_PIC}). Hence the constellation used for the scheme in \cite{2-USER_MIMO_MAC_PIC} should be of higher order to achieve the same spectral efficiency, which favours our scheme especially for high spectral efficiencies. It is worth noting that the higher rate of our two-user IC scheme comes at the expense of a lower coding gain compared to the two-user IC scheme in \cite{2-USER_MIMO_MAC_PIC}, consequently, for low-rates, our IC scheme may suffer from a performance loss w.r.t the latter scheme. However, when increasing the spectral efficiency, the loss due to the higher constellation size of the scheme in \cite{2-USER_MIMO_MAC_PIC} overrides its coding gain advantage, which favours our scheme. In Fig.~\ref{2x2x4_CER}, our rate-$5/7$ two-user IC scheme is compared to the L. shi {\it et al.} rate-$4/7$ two-user IC scheme in terms of CER in two cases. In the first case, the underlying constellations for our two-user IC scheme and the two-user IC scheme in \cite{2-USER_MIMO_MAC_PIC}, are 16-QAM and 32-QAM, respectively, whereas in the second case, the underlying constellations for our two-user IC scheme and the two-user IC scheme in \cite{2-USER_MIMO_MAC_PIC}, are 256-QAM and 1024-QAM, respectively. One can easily verify that the relative performance gain of the proposed two-user IC scheme w.r.t the L. shi {\it et al.} IC scheme increases with the spectral efficiency. The corresponding average decoding complexity measured in terms of average number of visited nodes is depicted in Fig.~\ref{2x2x4_AV_COM}. In addition to its superiority in terms of performance, our two-user IC scheme provides a substantial average decoding complexity reduction w.r.t the two-user IC scheme in \cite{2-USER_MIMO_MAC_PIC} especially in the low to average SNR regime.
%
%
%
%
\begin{figure}[h!]
\centering 
\begin{tikzpicture}

\begin{semilogyaxis}[%
view={0}{90},
width=4.0185in,
height=3.1694in,
scale only axis,
xmin=0, xmax=30,
xlabel={$\frac{E_b}{N_0}$ in dB},
xmajorgrids,
ymin=1e-07, ymax=1,
yminorticks=true,
ylabel={CER},
ymajorgrids,
yminorgrids,
legend style={at={(0.02,0.02)},anchor=south west,legend cell align=left}]
\addplot [
color=blue,
solid,
line width=2.0pt,
mark=o,
mark options={solid}
]
coordinates{
 (0,0.610088888888889)(4,0.193025555555556)(8,0.0226644444444444)(12,0.00109355555555561)(16,2.83333333333334e-05)(20,2.22222222222222e-07) 
};
\addlegendentry{\footnotesize Rate-5/7 IC scheme, 16-QAM};

\addplot [
color=red,
solid,
line width=2.0pt,
mark=square,
mark options={solid}
]
coordinates{
 (0,0.843451111111109)(4,0.441191666666666)(8,0.0809772222222221)(12,0.00388233333333344)(16,5.06666666666668e-05)(20,3.33333333333333e-07) 
};
\addlegendentry{\footnotesize Rate-4/7 IC scheme \cite{2-USER_MIMO_MAC_PIC}, 32-QAM};

\addplot [
color=blue,
solid,
line width=2.0pt,
mark=diamond,
mark options={solid}
]
coordinates{
 (0,0.999037777777782)(4,0.978516666666664)(8,0.802952222222221)(12,0.360967777777778)(16,0.0611122222222221)(20,0.00403355555555581)(24,0.000127999999999999)(28,2.22222222222222e-06) 
};
\addlegendentry{\footnotesize Rate-5/7 IC scheme, 256-QAM};

\addplot [
color=red,
solid,
line width=2.0pt,
mark=triangle,
mark options={solid}
]
coordinates{
 (0,0.999961111111112)(4,0.998978888888892)(8,0.982557777777779)(12,0.834450000000001)(16,0.387724444444445)(20,0.0561599999999999)(24,0.00210888888888899)(28,2.19999999999999e-05) 
};
\addlegendentry{\footnotesize Rate-4/7 IC scheme \cite{2-USER_MIMO_MAC_PIC}, 1024-QAM};

\end{semilogyaxis}
\end{tikzpicture}
\caption{Codeword error rate performance for the $\left(2^2,4\right)$ MIMO MAC channel.}
\label{2x2x4_CER}
\end{figure}
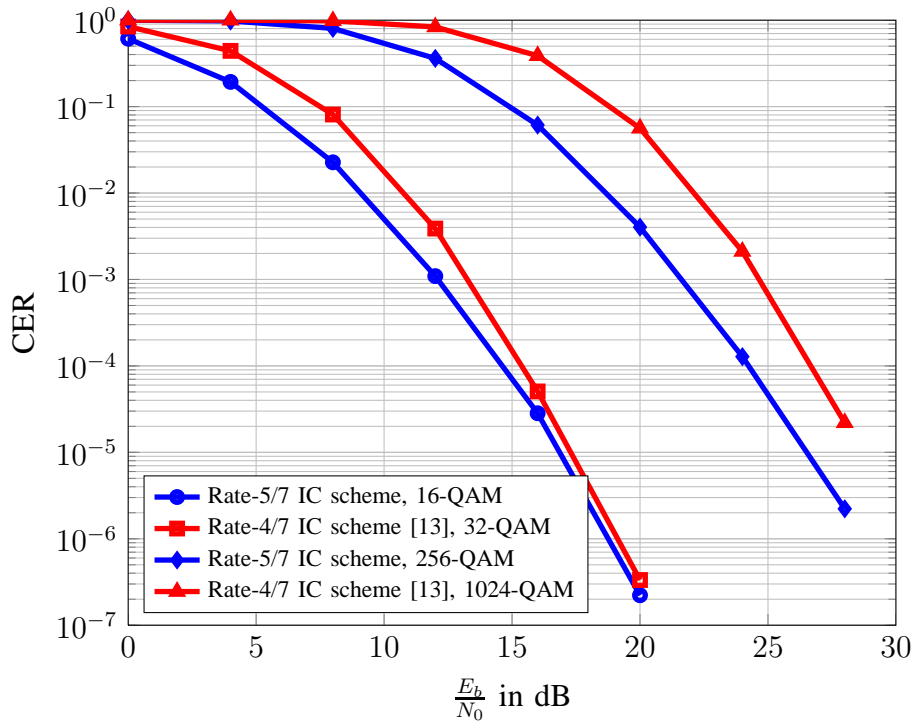
%
%
%
%
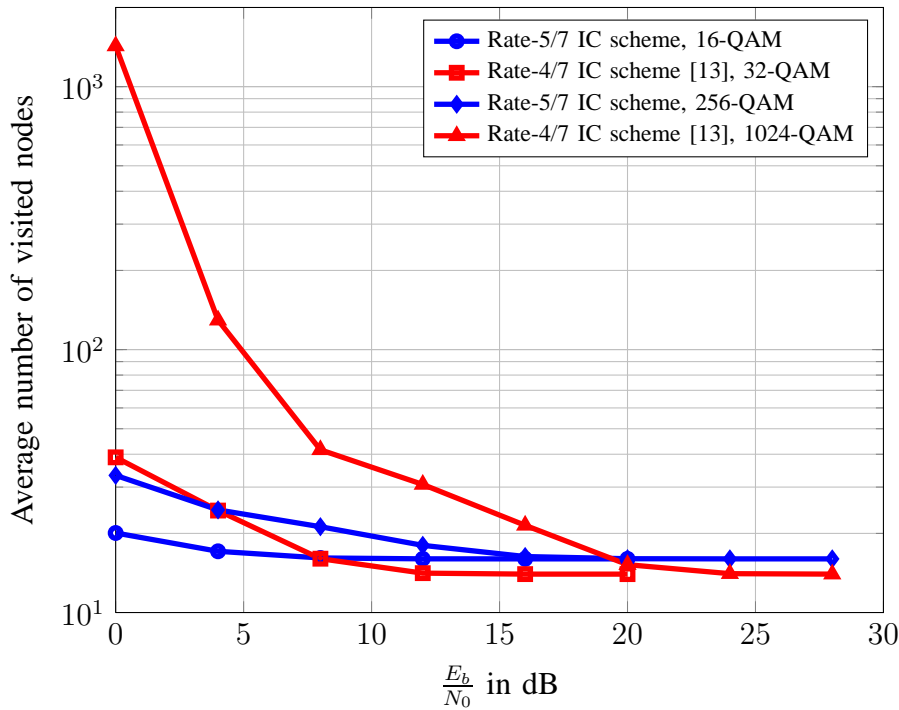
\begin{figure}[h!]
\centering 
\begin{tikzpicture}

\begin{semilogyaxis}[%
view={0}{90},
width=4.0185in,
height=3.1694in,
scale only axis,
xmin=0, xmax=30,
xlabel={$\frac{E_b}{N_0}$ in dB},
xmajorgrids,
ymin=10, ymax=2000,
yminorticks=true,
ylabel={Average number of visited nodes},
ymajorgrids,
yminorgrids,
legend style={at={(0.98,0.98)},anchor=north east,legend cell align=left}]
\addplot [
color=blue,
solid,
line width=2.0pt,
mark=o,
mark options={solid}
]
coordinates{
 (0,20.07204)(4,17.098545)(8,16.1289133333333)(12,16.0070665555555)(16,16.0002625555555)(20,16.0000075555556) 
};
\addlegendentry{\footnotesize Rate-5/7 IC scheme, 16-QAM};

\addplot [
color=red,
solid,
line width=2.0pt,
mark=square,
mark options={solid}
]
coordinates{
 (0,38.9294316666667)(4,24.4255166666667)(8,16.0257761111111)(12,14.1213365555556)(16,14.0021825555556)(20,14.0000123333333) 
};
\addlegendentry{\footnotesize Rate-4/7 IC scheme \cite{2-USER_MIMO_MAC_PIC}, 32-QAM};

\addplot [
color=blue,
solid,
line width=2.0pt,
mark=diamond,
mark options={solid}
]
coordinates{
 (0,33.2457644444444)(4,24.5693377777778)(8,21.2040066666667)(12,18.0140522222222)(16,16.33173)(20,16.0233453333334)(24,16.001024)(28,16.0000451111111) 
};
\addlegendentry{\footnotesize Rate-5/7 IC scheme, 256-QAM};

\addplot [
color=red,
solid,
line width=2.0pt,
mark=triangle,
mark options={solid}
]
coordinates{
 (0,1427.99527666667)(4,129.36933)(8,41.5846166666667)(12,30.6997755555556)(16,21.4649077777778)(20,15.2108584444445)(24,14.056486)(28,14.0007828888889) 
};
\addlegendentry{\footnotesize Rate-4/7 IC scheme \cite{2-USER_MIMO_MAC_PIC}, 1024-QAM};

\end{semilogyaxis}
\end{tikzpicture}
\caption{Average decoding complexity performance for the $\left(2^2,4\right)$ MIMO MAC channel.}
\label{2x2x4_AV_COM}
\end{figure}


\subsection*{Worst-Case Decoding Complexity Analysis}
The worst-case decoding complexity order is defined by the minimum number of times the Euclidean distance metric has to be evaluated for the optimal decoder to estimate the transmitted codeword \cite{FAST}. In view of the aforementioned definition, our proposed IC scheme offers significant worst-case decoding complexity reduction as compared to \cite{2-USER_MIMO_MAC_PIC}. In order to ensure equal signalling periods, $n'$ is chosen such that $n=n'+N_t-1$. Let $q$ (resp. $q'$) denote the size of the underlying QAM constellation for our IC scheme (resp. the IC scheme in \cite{2-USER_MIMO_MAC_PIC}). Therefore, for square QAM constellations the worst-case decoding complexity for our scheme is $\Oc\left(q^{n/2}\right)$ thanks to the separate decoding of real and imaginary parts of each user's symbols, whereas for the L. shi {\it et al.} scheme \cite{2-USER_MIMO_MAC_PIC}, the worst-case decoding complexity is $\Oc\left(q'^{n'}\right)$. On the other hand, for equal spectral efficiencies, one has $q=q'^{n'/n}$, which yields a worst-case decoding complexity of $\Oc\left(q'^{n'/2}\right)$ for the proposed IC scheme. It is worth noting that in order to preserve the low-complexity of detection for our IC scheme, the underlying QAM constellations should be rectangular to avoid any dependence between the real and imaginary parts of the transmitted symbols. Restricting the underlying constellations to be rectangular is far from being deleterious, this is due to the fact that rectangular constellations provide generally a higher constellation figure of merit (defined as the ratio of the minimum squared distance between any two points of the constellation to its average power) than non-rectangular constellations (e.g., PSK and HEX), thus providing better performance. Moreover, QAM constellations are adopted in recent wireless communication standards such as LTE and LTE advanced, thus restricting ourselves to rectangular constellations is in compliance with current and upcoming norms.


\section{Conclusion}
In this paper we focused on the multi-user MIMO MAC configuration and provided new sufficient conditions for a large family of STBCs to achieve the full-diversity gain offered by the channel under PICGD in the absence of CSIT. Explicit IC schemes that satisfy the full-diversity criteria were then proposed for two and three user MIMO MAC. We proved that the proposed IC schemes enable further reduction of the decoding complexity through separate decoding of real and imaginary parts of each user's transmitted symbols without any loss of performance. The proposed schemes were then compared to their counterpart in the literature and found to outperform their rival in the literature especially at high spectral efficiencies, while having a significantly less decoding complexity. 


\section{Acknowldgment}
This work was partially funded by a grant from King Abdul-Aziz City of Science and Technology (KACST).


\section*{Appendix A}
\begin{proof}
The general solution to Eq.~\eqref{cond2} may be expressed as
\begin{equation*}
\Cm^\mathsf{H}=\Qm\Mm;\ \Mm=\left(\Id-\Am\Am^{\dagger}\right) 
\end{equation*}
which implies that
\begin{equation}
\text{r}\left(\Cm^\mathsf{H}\Vm\Cm\right)= \text{r}\left(\Cm^\mathsf{H}\Vm^{\frac{1}{2}}\right)=\text{r}\left(\Qm\Mm\Vm^{\frac{1}{2}}\right)\overset{(b)}{\leq} \text{r}\left(\Mm\Vm^{\frac{1}{2}}\right)= \text{r}\left(\Mm\Vm\Mm\right),  
\label{upbnd1}
\end{equation}
where $(b)$ follows from $\text{r}\left(\Xm\Ym\right)\leq\min\left\lbrace\text{r}\left(\Xm\right),\text{r}\left(\Ym\right)\right\rbrace$.
Thanks to the Frobenius rank inequality \cite{MATRIX_ANALYSIS}, one has
\begin{equation*}
\text{r}\left(\Xm\Ym\Zm\right)+\text{r}\left(\Ym\right)\geq\text{r}\left(\Xm\Ym\right)+\text{r}\left(\Ym\Zm\right)
\end{equation*}
for arbitrary matrices $\Xm\in\mathbb{C}^{k\times l},\Ym\in\mathbb{C}^{l\times m},\Zm\in\mathbb{C}^{m\times n}$.
A straightforward application of the above inequality results in
\begin{equation}
\text{r}\left(\Qm\Mm\Vm^{\frac{1}{2}}\right)+\text{r}\left(\Mm\right)\geq\text{r}\left(\Qm\Mm\right)+\text{r}\left(\Mm\Vm^{\frac{1}{2}}\right).
\label{Frobenius_ineq}
\end{equation}
However from Eq.~\eqref{cond1} one has
\begin{equation*}
\text{r}\left(\Qm\Mm\right)=\text{r}\left(\Cm^\mathsf{H}\right)=\text{r}\left(\Cm\right)=p-\text{r}\left(\Am\right)\overset{(c)}{=}\text{r}\left(\Mm\right),
\end{equation*}
where $(c)$ follows from the definition of $\Mm$. Consequently, the inequality in \eqref{Frobenius_ineq} reduces to
\begin{equation}
\text{r}\left(\Qm\Mm\Vm^{\frac{1}{2}}\right)\geq\text{r}\left(\Mm\Vm^{\frac{1}{2}}\right)=\text{r}\left(\Mm\Vm\Mm\right).
\label{lrbnd1}
\end{equation}
Combining \eqref{upbnd1} and \eqref{lrbnd1} one obtains
\begin{equation*}
\text{r}\left(\Cm^\mathsf{H}\Vm\Cm\right)=\text{r}\left(\Mm\Vm\Mm\right).
\end{equation*}
At this stage, we need to prove that
\begin{equation*}
\text{r}\left(\begin{bmatrix}\Am & \Vm \end{bmatrix}\right)=\text{r}\left(\Mm\Vm\Mm\right)+\text{r}\left(\Am\right)   .
\end{equation*}
Towards this end, let us define $\Mm_c=\Am\Am^{\dagger}$, thus we have
\begin{eqnarray}
\text{r}\left(\begin{bmatrix}\Am &\Vm \end{bmatrix}\right)&\overset{(d)}{=}&\text{r}\left(\begin{bmatrix}\zerov & \Mm\Vm \end{bmatrix}+\begin{bmatrix}\Am & \Mm_c\Vm \end{bmatrix}\right)\\
&=&\text{r}\left(\begin{bmatrix}\zerov & \Mm\Vm \end{bmatrix}+\Am\begin{bmatrix}\Id_q & \Am^{\dagger}\Vm \end{bmatrix}\right),
\label{Rank_eq1}
\end{eqnarray}
where $(d)$ follows from the definition of $\Mm_c$. On the other hand, one has
\begin{equation*}
\text{r}\left(\Am \begin{bmatrix} \Id_q &\Am^{\dagger}\Vm \end{bmatrix}\right)\overset{(e)}{=}\text{r}\left(\begin{bmatrix} \Id_q& \Am^{\dagger}\Vm \end{bmatrix}\right)=q=\text{r}\left(\Am\right),   
\end{equation*} 
where $(e)$ follows from the assumption that $\Am$ is of full column rank and the rank equality $\text{r}\left(\Xm\Ym\right)=\text{r}\left(\Ym\right)$ for $\Xm\in\mathbb{C}^{k\times l},\Ym \in\mathbb{C}^{l\times m}$ if $\Xm$ is of rank $l$. Therefore
\begin{equation*}
\Mc\left(\Am \begin{bmatrix} \Id_q &\Am^{\dagger}\Vm \end{bmatrix}\right)=\Mc\left(\Am\right).
\end{equation*} 
Consequently, Eq~\eqref{Rank_eq1} may be expressed as
\begin{equation}
\begin{split}
\text{r}\left(\begin{bmatrix}\Am &\Vm \end{bmatrix}\right)
&=\text{r}\left(\begin{bmatrix}\zerov & \Mm\Vm \end{bmatrix}+\begin{bmatrix}\Am & \zerov\end{bmatrix}\right)\\
&=\text{r}\left(\begin{bmatrix}\Am & \Mm\Vm \end{bmatrix}\right)                  
\end{split}
\label{Rank_eq2}
\end{equation}
The above result above is easily verified as $\Mm$ can be interpreted as the projection matrix into a columns subspace orthogonal to $\Mc\left(\Am\right)$. Noting that $\Vm^{\frac{1}{2}}\succeq 0$, and thanks to the identity $\Mc\left(\Xm\right)=\Mc\left(\Xm\Xm^\mathsf{H}\right)$ for arbitrary matrix $\Xm \in \mathbb{C}^{m\times n}$ \cite{MATRIX_CALC}, one has $\Mc\left(\Vm^{\frac{1}{2}}\right)=\Mc\left(\Vm\right)$. Therefore, Eq.~\eqref{Rank_eq2} may be rewritten in the following form
\begin{equation*}
\text{r}\left(\begin{bmatrix}\Am &\Vm \end{bmatrix}\right)=\text{r}\left(\begin{bmatrix}\Am &\Vm^{\frac{1}{2}} \end{bmatrix}\right)=\text{r}\left(\begin{bmatrix}\Am &\Mm\Vm^{\frac{1}{2}} \end{bmatrix}\right). 
\end{equation*}
Moreover, one has
\begin{equation*}
\begin{split}
\text{r}\left(\begin{bmatrix}\Am &\Mm\Vm^{\frac{1}{2}} \end{bmatrix}\right)&=\text{r}\left(\begin{bmatrix}\Am^\mathsf{H} \\ \Vm^{\frac{1}{2}}\Mm \end{bmatrix}\begin{bmatrix}\Am &\Mm\Vm^{\frac{1}{2}} \end{bmatrix}\right)\\
&=\text{r}\left(\begin{bmatrix}\Am^\mathsf{H}\Am & \zerov\\ \zerov & \Vm^{\frac{1}{2}}\Mm^2\Vm^{\frac{1}{2}}\end{bmatrix}\right)\\
&=\text{r}\left(\Am^\mathsf{H}\Am\right)+\text{r}\left(\Vm^{\frac{1}{2}}\Mm^2\Vm^{\frac{1}{2}}\right)\\
&\overset{(f)}{=}\text{r}\left(\Am\right)+\text{r}\left(\Mm\Vm\Mm\right),
\end{split}
\end{equation*}
where $(f)$ follows from the identity $\text{r}\left(\Xm\Xm^\mathsf{H}\right)=\text{r}\left(\Xm^\mathsf{H}\Xm\right)$, thus completing the proof.
\end{proof}


\section*{Appendix B}
\begin{proof}
Consider the case of one receive antenna at the destination, according to \eqref{2_user_IC} and \eqref{C}, one has
\begin{equation}
\Hcb_{1,1}\left(\hv_{1,1}\right)=\begin{bmatrix}\Dm\left(\hv_{1,1}\right)\\\underset{1\times n}{\zerov}\end{bmatrix},
\Hcb_{2,1}\left(\hv_{2,1}\right)=\begin{bmatrix}\underset{1\times n}{\zerov}\\\Dm\left(\hv_{2,1}\right)\end{bmatrix}
\label{2_user_Heq}
\end{equation}
where $\Dm\left(\hv\right)$ is as defined in \eqref{D}. Thanks to the column-wise orthogonality, it is straightforward to verify that
\begin{equation}
\Dm\left(\hv_{k,1}\right)^\mathsf{H}\Dm\left(\hv_{k,1}\right)=
\text{diag}\left(\begin{bmatrix}\Vert\dv_{k,1}\Vert^2&\ldots&\Vert\dv_{k,n}\Vert^2\end{bmatrix}\right)
\label{orth}
\end{equation}
where $\dv_{k,i}$ denotes the $i$-th column of the matrix $\Dm\left(\hv_{k,1}\right)$.
Suppose without loss of generality that the first user is being decoded, the optimal detection rule with PICGD turns into
\begin{equation*}
\sv^{\text{ML$\mid$PICGD}}_1=\text{arg}\underset{\sv_1\in\Ac_1}{\min}\Vert\Pm_1\yv-\Pm_1\Hcb_{1,1}\Um_n\sv_1\Vert.
\end{equation*}
For the proposed two-user IC scheme \eqref{2_user_IC}, it can be easily checked that $\Hcb_{\overline{1}}=\Hcb_{2,1}$, therefore
\begin{equation}
\begin{split}
\Um_n^\mathsf{T}\Hcb_{1,1}^\mathsf{H}\Pm_1\Hcb_{1,1}\Um_n&=
\Um_n^\mathsf{T}\Hcb_{1,1}^\mathsf{H}\left(\Id-\Hcb_{2,1}\left(\Hcb_{2,1}^\mathsf{H}\Hcb_{2,1}\right)^{-1}\Hcb_{2,1}^\mathsf{H}\right)\Hcb_{1,1}\Um_n\\
&=\Um_n^\mathsf{T}\left(\Hcb_{1,1}^\mathsf{H}\Hcb_{1,1}-\Hcb_{1,1}^\mathsf{H}\Hcb_{2,1}\left(\Hcb_{2,1}^\mathsf{H}\Hcb_{2,1}\right)^{-1}\Hcb_{2,1}^\mathsf{H}\Hcb_{1,1}\right)\Um_n.
\end{split}
\label{2_user_Cholesky}
\end{equation}
From \eqref{2_user_Heq} and \eqref{orth} we have that $\Hcb_{1,1}^\mathsf{H}\Hcb_{1,1}$ and $\left(\Hcb_{2,1}^\mathsf{H}\Hcb_{2,1}\right)^{-1}$ are both diagonal matrices with real entries. On the other hand, it can be easily verified from \eqref{D} that $\Hcb_{1,1}^\mathsf{H}\Hcb_{2,1}=\Lambda_{12}\Pm_{\pmb{\sigma}_{12}}$, where $\pmb{\sigma}_{12}=\begin{bmatrix}2&3&\ldots&n&1\end{bmatrix}$ and $\pmb{\Lambda}_{12}$ is defined as 
\begin{equation*}
\pmb{\Lambda}_{12}=\text{diag}\left(
\begin{bmatrix}\left<\dv_1\left(\hv_{1,1}\right),
\dv_n\left(\hv_{2,1}\right)\right>&\left<\dv_2\left(\hv_{1,1}\right),
\dv_1\left(\hv_{2,1}\right)\right>&\ldots&\left<\dv_n\left(\hv_{1,1}\right),
\dv_{n-1}\left(\hv_{2,1}\right)\right>\end{bmatrix}\right).
\end{equation*}
Recall that for a diagonal matrix $\pmb{\Lambda}$, $\Pm_{\pmb{\sigma}}\pmb{\Lambda}\Pm_{\pmb{\sigma}}^\mathsf{T}$ is also diagonal \cite{MATRIX_ANALYSIS}, consequently the matrix $\Hcb_{1,1}^\mathsf{H}\Hcb_{2,1}\left(\Hcb_{2,1}^\mathsf{H}\Hcb_{2,1}\right)^{-1}\Hcb_{2,1}^\mathsf{H}\Hcb_{1,1}$ is diagonal with real entries. Therefore choosing $\Um_n$ to be real implies that $\Um_n^\mathsf{T}\Hcb_{1,1}^\mathsf{H}\Pm_1\Hcb_{1,1}\Um_n$ is real. But $\Um_n^\mathsf{T}\Hcb_{1,1}^\mathsf{H}\Pm_1\Hcb_{1,1}\Um_n\succeq 0$, thus may be factored according to the Cholesky decomposition as $\Lm\Lm^\mathsf{H}$ with $\Lm$ being a real lower triangular matrix \cite{MATRIX_ANALYSIS}. Applying the QR-decomposition, one obtains $\Pm_1\Hcb_{1,1}\Um_n=\Qm\Rm$, thus $\Um_n^\mathsf{T}\Hcb_{1,1}^\mathsf{H}\Pm_1\Hcb_{1,1}\Um_n=\Lm\Lm^\mathsf{H}=\Rm^\mathsf{H}\Rm\Rightarrow \Rm\in\mathbb{R}^{T\times n}$. The ML decision rule under PICGD reduces to
\begin{IEEEeqnarray*}{c'c'c}
\Re\left(\sv_1\right)^{\text{ML}\mid\text{PICGD}}&=&\text{arg}\underset{\hat{\xv}\in
\Re\left\lbrace\Ac_1\right\rbrace}{\min}\Big\Vert\Re\left\lbrace\Qm_1^\mathsf{H}\Pm_1\yv\right\rbrace-\Rm_1\hat{\xv}\Big\Vert^2\\
\Im\left(\sv_1\right)^{\text{ML}\mid\text{PICGD}}&=&\text{arg}\underset{\hat{\xv}\in
\Im\left\lbrace\Ac_1\right\rbrace}{\min}\Big\Vert\Im\left\lbrace\Qm_1^\mathsf{H}\Pm_1\yv\right\rbrace-\Rm_1\hat{\xv}\Big\Vert^2,
\end{IEEEeqnarray*}
where $\Qm=\begin{bmatrix}\Qm_1 & \Qm_2\end{bmatrix},\ \Rm=\begin{bmatrix}\Rm^\mathsf{T}_1&\zerov \end{bmatrix}^\mathsf{T}$. The same approach can be adopted to prove the separability of the real and imaginary parts of $\sv_2$ and is therefore omitted. The case of arbitrary number of receive antennas follows similarly, thus ending the proof.
\end{proof}


\section*{Appendix C}
\begin{proof}
Consider the case of one receive antenna at the destination, according to \eqref{3_user_IC} and \eqref{C}, one has the following
\begin{equation}
\Hcb_{1,1}=\begin{bmatrix}\Dm\left(\hv_{1,1}\right)\\\underset{\left(n+1\right)\times n}{\zerov}\end{bmatrix},\
\Hcb_{2,1}=\begin{bmatrix}\underset{N_t\times n}{\zerov}\\\Dm\left(\hv_{2,1}\right)\\\underset{n-N_t+1}{\zerov}\end{bmatrix},\
\Hcb_{3,1}=\begin{bmatrix}\underset{\left(n+1\right)\times n}{\zerov}\\\Dm\left(\hv_{3,1}\right)\end{bmatrix}.
\label{3_user_Heq}
\end{equation}
We proceed by decoding the first user
\begin{equation*}
\sv^{\text{ML$\mid$PICGD}}_1=\text{arg}\underset{\sv_1\in\Ac_1}{\min}\Vert\Pm_1\yv-\Pm_1\Hcb_{1,1}\Um_n\sv_1\Vert.
\end{equation*}
For the proposed two-user IC scheme \eqref{3_user_IC}, it can be easily checked that $\Hcb_{\overline{1}}=\begin{bmatrix}\Hcb_{2,1}&\Hcb_{3,1}\end{bmatrix}$, therefore
\begin{equation}
\Um_n^\mathsf{T}\Hcb_{1,1}^\mathsf{H}\Pm_1\Hcb_{1,1}\Um_n=
\Um_n^\mathsf{T}\Hcb_{1,1}^\mathsf{H}\left(\Id-\begin{bmatrix}\Hcb_{2,1}&\Hcb_{3,1}\end{bmatrix}
\begin{bmatrix}\Hcb_{2,1}^\mathsf{H}\Hcb_{2,1}&\Hcb_{2,1}^\mathsf{H}\Hcb_{3,1}\\\Hcb_{3,1}^\mathsf{H}\Hcb_{2,1}&\Hcb_{3,1}^\mathsf{H}\Hcb_{3,1}\end{bmatrix}^{-1}
\begin{bmatrix}\Hcb_{2,1}^\mathsf{H}\\\Hcb_{3,1}^\mathsf{H}\end{bmatrix}\right)\Hcb_{1,1}\Um_n.
\label{3_user_Cholesky}
\end{equation}
Thanks to \eqref{3_user_Heq} and \eqref{D} we have
\begin{IEEEeqnarray*}{c;c;l}
\Hcb_{i,1}^\mathsf{H}\Hcb_{i,1}&=&\pmb{\Lambda}_{ii},\ \forall\ i=1,2,3,\\
\Hcb_{i,1}^\mathsf{H}\Hcb_{j,1}&=&\pmb{\Lambda}_{ij}\Pm_{\pmb{\sigma}_{ij}},\forall\ 1\leq i<j\leq 3,
\end{IEEEeqnarray*}
where 
\begin{IEEEeqnarray*}{c;c;l}
\pmb{\sigma}_{12}&=&\begin{bmatrix}N_t+1&\ldots&n&1&\ldots&N_t\end{bmatrix},\\
\pmb{\sigma}_{13}&=&\begin{bmatrix}2&\ldots&\ldots&n&1\end{bmatrix},\\
\pmb{\sigma}_{23}&=&\begin{bmatrix}n-N_t+2&\ldots&n&1&\ldots&n-N_t+1\end{bmatrix},\\
\pmb{\Lambda}_{ii}&=&\text{diag}
\left(\begin{bmatrix}\Vert\dv_{i,1}\Vert^2&\ldots&\Vert\dv_{i,n}\Vert^2\end{bmatrix}\right),\forall\ i=1,2,3,\\
\pmb{\Lambda}_{12}&=&\text{diag}\left(
\begin{bmatrix}\left<\dv_{1,1},\dv_{2,n-N_t+1}\right>&\ldots&\left<\dv_{1,N_t-1},\dv_{2,n-1}\right>&0&\left<\dv_{1,N_t+1},\dv_{2,1}\right>&\ldots&
\left<\dv_{1,n},\dv_{2,n-N_t}\right>\end{bmatrix}\right)\\
\pmb{\Lambda}_{13}&=&\text{diag}\left(
\begin{bmatrix}0&\left<\dv_{1,2},\dv_{3,1}\right>&\ldots&\left<\dv_{1,N_t-1},\dv_{3,N_t-2}\right>&0&\ldots&0\end{bmatrix}\right)\\
\pmb{\Lambda}_{23}&=&\text{diag}\left(
\begin{bmatrix}\left<\dv_{2,1},\dv_{3,N_t}\right>&\ldots&\left<\dv_{2,N_t-1},\dv_{3,2N_t-2}\right>&0&\ldots&0&\left<\dv_{2,n-N_t+2},\dv_{3,1}\right>&\ldots&
\left<\dv_{2,n},\dv_{3,N_t-1}\right>\end{bmatrix}\right)
\end{IEEEeqnarray*}
Consequently, \eqref{3_user_Cholesky} may be re-written as
\begin{equation*}
\begin{split}
&\Um_n^\mathsf{T}\left(\pmb{\Lambda}_{11}-\begin{bmatrix}\pmb{\Lambda}_{12}\Pm_{\pmb{\sigma}_{12}}&\pmb{\Lambda}_{13}\Pm_{\pmb{\sigma}_{13}}\end{bmatrix}
\begin{bmatrix}\pmb{\Lambda}_{22}&\pmb{\Lambda}_{23}\Pm_{\pmb{\sigma}_{23}}\\
\Pm_{\pmb{\sigma}_{23}}^\mathsf{T}\pmb{\Lambda}_{23}^\mathsf{H}&\pmb{\Lambda}_{33}\end{bmatrix}^{-1}
\begin{bmatrix}\Pm_{\pmb{\sigma}_{12}}^\mathsf{T}\pmb{\Lambda}_{12}^\mathsf{H}\\\Pm_{\pmb{\sigma}_{13}}^\mathsf{T}\pmb{\Lambda}_{13}^\mathsf{H}\end{bmatrix}\right)\Um_n\\
&\overset{(g)}{=}\Um_n^\mathsf{T}\left(\pmb{\Lambda}_{11}-\begin{bmatrix}\pmb{\Lambda}_{12}\Pm_{\pmb{\sigma}_{12}}&\pmb{\Lambda}_{13}\Pm_{\pmb{\sigma}_{13}}\end{bmatrix}
\begin{bmatrix}\pmb{\Sigma}_{11}&-\pmb{\Lambda}_{22}^{-1}\pmb{\Lambda}_{23}\Pm_{\pmb{\sigma}_{23}}\pmb{\Sigma}_{22}\\
-\pmb{\Lambda}_{33}^{-1}\Pm_{\pmb{\sigma}_{23}}^\mathsf{T}\pmb{\Lambda}_{23}^\mathsf{H}\pmb{\Sigma}_{11}&\pmb{\Sigma}_{22}\end{bmatrix}
\begin{bmatrix}\Pm_{\pmb{\sigma}_{12}}^\mathsf{T}\pmb{\Lambda}_{12}^\mathsf{H}\\\Pm_{\pmb{\sigma}_{13}}^\mathsf{T}\pmb{\Lambda}_{13}^\mathsf{H}\end{bmatrix}\right)\Um_n\\
&=\Um_n^\mathsf{T}(\pmb{\Lambda}_{11}-
\underbrace{\pmb{\Lambda}_{12}\Pm_{\pmb{\sigma}_{12}}\pmb{\Sigma}_{11}\Pm_{\pmb{\sigma}_{12}}^\mathsf{T}\pmb{\Lambda}_{12}^\mathsf{H}}_{\Am}
+\underbrace{\pmb{\Lambda}_{13}\Pm_{\pmb{\sigma}_{13}}\pmb{\Lambda}_{33}^{-1}\Pm_{\pmb{\sigma}_{23}}^\mathsf{T}\pmb{\Lambda}_{23}^\mathsf{H}\pmb{\Sigma}_{11}\Pm_{\pmb{\sigma}_{12}}^\mathsf{T}\pmb{\Lambda}_{12}^\mathsf{H}}_{\Cm}
-\underbrace{\pmb{\Lambda}_{13}\Pm_{\pmb{\sigma}_{13}}\pmb{\Sigma}_{22}\Pm_{\pmb{\sigma}_{13}}^\mathsf{T}\pmb{\Lambda}_{13}^\mathsf{H}}_{\Bm}\\
&+\underbrace{\pmb{\Lambda}_{12}\Pm_{\pmb{\sigma}_{12}}\pmb{\Lambda}_{22}^{-1}\pmb{\Lambda}_{23}\Pm_{\pmb{\sigma}_{23}}\pmb{\Sigma}_{22}\Pm_{\pmb{\sigma}_{13}}^\mathsf{T}\pmb{\Lambda}_{13}^\mathsf{H}}_{\Cm^\mathsf{H}}
)\Um_n
\end{split}
\end{equation*}
where $(g)$ follows from the blockwise matrix inverse \cite{MATRIX_ANALYSIS}, $\pmb{\Sigma}_{11}=\left(\pmb{\Lambda}_{22}-\pmb{\Lambda}_{23}\Pm_{\pmb{\sigma}_{23}}\pmb{\Lambda}_{33}^{-1}\Pm_{\pmb{\sigma}_{23}}^\mathsf{T}\pmb{\Lambda}_{23}^\mathsf{H}\right)^{-1}$,  and $\pmb{\Sigma}_{22}=\left(\pmb{\Lambda}_{33}-\Pm_{\pmb{\sigma}_{23}}^\mathsf{T}\pmb{\Lambda}_{23}^\mathsf{H}\pmb{\Lambda}_{22}^{-1}\pmb{\Lambda}_{23}\Pm_{\pmb{\sigma}_{23}}\right)^{-1}$               . Recall that for a diagonal matrix $\pmb{\Lambda}$, $\Pm_{\pmb{\sigma}}\pmb{\Lambda}\Pm_{\pmb{\sigma}}^\mathsf{T}$ is also diagonal \cite{MATRIX_ANALYSIS}, hence $\pmb{\Sigma}_{11}$ and $\pmb{\Sigma}_{22}$ are diagonal matrices with real entries which in turns implies that $\Am$ and $\Bm$ are diagonal matrices with real entries. On the other hand, it is well known that the product of two permutation matrices $\Pm_{\pmb{\sigma}}\Pm_{\pmb{\epsilon}}$ is given by $\Pm_{\pmb{\pi}}$ such that $\pmb{\pi}=\pmb{\sigma}\circ\pmb{\epsilon}$. Accordingly, it can be verified that $\Pm_{\pmb{\sigma}_{12}}\Pm_{\pmb{\sigma}_{23}}=\Pm_{\pmb{\sigma}_{13}}$, thus yielding $\Pm_{\pmb{\sigma}_{12}}\Pm_{\pmb{\sigma}_{23}}\Pm_{\pmb{\sigma}_{13}}^\mathsf{T}=\Id$. Consequently, $\Cm$ is diagonal with complex entries, and therefore for real rotation matrices $\Um_n$ one has that $\Um_n^\mathsf{T}\Hcb_{1,1}^\mathsf{H}\Pm_1\Hcb_{1,1}\Um_n$ is real. But $\Um_n^\mathsf{T}\Hcb_{1,1}^\mathsf{H}\Pm_1\Hcb_{1,1}\Um_n\succeq 0$, thus may be factored according to the Cholesky decomposition as $\Lm\Lm^\mathsf{H}$ with $\Lm$ being a real lower triangular matrix \cite{MATRIX_ANALYSIS}. Applying the QR-decomposition, one obtains $\Pm_1\Hcb_{1,1}\Um_n=\Qm\Rm$, thus $\Um_n^\mathsf{T}\Hcb_{1,1}^\mathsf{H}\Pm_1\Hcb_{1,1}\Um_n=\Lm\Lm^\mathsf{H}=\Rm^\mathsf{H}\Rm\Rightarrow \Rm\in\mathbb{R}^{T\times n}$. The ML decision rule under PICGD reduces to
\begin{IEEEeqnarray*}{c;c;c}
\Re\left(\sv_1\right)^{\text{ML}\mid\text{PICGD}}&=&\text{arg}\underset{\hat{\xv}\in
\Re\left\lbrace\Ac_1\right\rbrace}{\min}\Big\Vert\Re\left\lbrace\Qm_1^\mathsf{H}\Pm_1\yv\right\rbrace-\Rm_1\hat{\xv}\Big\Vert^2\\
\Im\left(\sv_1\right)^{\text{ML}\mid\text{PICGD}}&=&\text{arg}\underset{\hat{\xv}\in
\Im\left\lbrace\Ac_1\right\rbrace}{\min}\Big\Vert\Im\left\lbrace\Qm_1^\mathsf{H}\Pm_1\yv\right\rbrace-\Rm_1\hat{\xv}\Big\Vert^2,
\end{IEEEeqnarray*}
where $\Qm=\begin{bmatrix}\Qm_1 & \Qm_2\end{bmatrix},\ \Rm=\begin{bmatrix}\Rm^\mathsf{T}_1&\zerov \end{bmatrix}^\mathsf{T}$. It is straightforward to prove that the real and imaginary parts of $\sv_2$ and $\sv_3$ can be decoded separately without any loss of performance. The case of arbitrary number of receive antennas follows similarly, thus ending the proof.
\end{proof}


\ifwindows
\bibliography{C:/Users/amr.ismail/Desktop/Submissions/BibTeX/STBCs,C:/Users/amr.ismail/Desktop/Submissions/BibTeX/MU-STBCs,C:/Users/amr.ismail/Desktop/Submissions/BibTeX/PIC,E:/BibTeX/Books}
\else
\bibliography{/media/WORK/BibTeX/STBCs,/media/WORK/BibTeX/MU-STBCs,/media/WORK/BibTeX/    PIC,/media/WORK/BibTeX/Books}
\fi
\bibliographystyle{ieeetr}
\end{document}